\documentclass[a4paper,UKenglish,cleveref, autoref, thm-restate]{lipics-v2021}

\pdfoutput=1 
\hideLIPIcs  
\nolinenumbers

\usepackage{adjustbox,rotating,booktabs,tcolorbox} 

\usepackage{quiver}
\usepackage{etoolbox}
\usepackage{fullshort}
\setboolean{fullpaper}{true}

\newcommand\tom[1]{}
\newcommand\tominline[1]{}

\usepackage[toc,page]{appendix}
\usepackage{savesym}
\savesymbol{widering}
\usepackage{xspace,mathpartir,xparse,url,multirow,wrapfig}
\usepackage{thmtools,thm-restate}
\usepackage{ebmath,ebproof,ebutf8}
\usepackage{listings}

\newcommand{\isempty}[1]%
{
  \ifthenelse{\equal{#1}{}}%
    {EMPTY}
    {FULL, it contains the string '#1'}
}

\usepackage{environ}

\usepackage{tikz-cats}
\usetikzlibrary{matrix,cd}
\tikzcdset{arrow style=tikz}
\tikzcdset{every arrow/.append style = -{Computer Modern Rightarrow[]}}
\tikzset{every arrow/.append style = -{Computer Modern Rightarrow[]}}
\tikzset{
  labelslt/.style n args={2}{%
    labelat={[left]{$\scriptstyle #1$}}{.45},%
    labelat={[right]{$\scriptstyle #2$}}{.55}%
    } %
  }
\tikzset{
  labelsltat/.style n args={4}{%
    labelat={[left]{$\scriptstyle #1$}}{#3},%
    labelat={[right]{$\scriptstyle #2$}}{#4}%
    } %
  }
\tikzset{twoof/.style={twocenter={#1}}}

\newcounter{trou}[figure]
\renewcommand{\diaggrandhauteur}{.6}
\renewcommand{\diaggrandlargeur}{1}

\newcommand{\psh}[1][ℂ]{\widehat{#1}}
\newcommand{\xto}[1]{\xrightarrow{#1}}
\newcommand{\xinto}[1]{\xarrow[into]{#1}}
\newcommand{\xonto}[1]{\xarrow[onto]{#1}}

\newcommand{\xot}[1]{\xleftarrow{#1}}

\newcommand{\restr}[2]{{#1}_{{|}#2}}
\newcommand{\mrestr}[2]{{#1}_{{↾}#2}}

\newcommand{\wbotright}[1]{{#1}^⋔}
\newcommand{\wbotleft}[1]{{{}^⋔{#1}}}
\newcommand{\wbotrightleft}[1]{\wbotleft{(\wbotright{#1})}}
\newcommand{\source}{𝐬}
\newcommand{\labels}{𝐥}
\newcommand{\but}{𝐭}
\newcommand{\op}[1]{{#1}^{\mathit{op}}}

\DeclareMathOperator{\alg}{-\,\mathbf{alg}}
\DeclareMathOperator{\Alg}{-\,\mathbf{Alg}}

\DeclareMathOperator{\Mod}{-\,\mathbf{Mod}}

\DeclareMathOperator{\Trans}{-\mathbf{Trans}}

\DeclareMathOperator{\id}{id}
\DeclareMathOperator{\el}{el}
\DeclareMathOperator{\colim}{colim}

\DeclareMathOperator{\im}{im}

\DeclareTextMath\tmlambda[lambda]{\lambda}

\usepackage{bbding}

\newcommand{\iso}{≅}

\newcommand{\ajustedroit}[2][1]{\adjustbox{max width=#1\columnwidth,max height=.95\textheight}{#2}}%
\newcommand{\alert}[1]{\textbf{#1}}

\NewDocumentCommand{\GammaInv}{}{\mathpalette\doGammaInv\relax}
\makeatletter
\NewDocumentCommand{\doGammaInv}{mm}{%
  \reflectbox{$\m@th#1\Gamma$}%
}
\makeatother
\renewcommand{\daleth}{\GammaInv}

\newcommand{\shift}{\mathrm{shift}}

\newtheorem{terminology}{Terminology}

\definecolor{bois}{rgb}{.5,0,0}
\definecolor{vert}{rgb}{0,0.6,0}
\definecolor{rouge}{rgb}{0.8,0,0}
\definecolor{violet}{rgb}{0.8,0,0.4}
\definecolor{marron}{rgb}{0.4,0.4,0}
\definecolor{bleu}{rgb}{0,0,0.6}
\AtBeginDocument
{

}

\newcommand{\enhanced}{enhanced\xspace}
\newcommand{\Enhanced}{Enhanced\xspace}

\newcommand{\anenhanced}{an enhanced\xspace}

\newcommand{\Fwow}[1][X]{Σ₀^{{+};{∼_{#1}}}}



\DeclareUnicodeMathCharacter{21A4}{\mathrel{\xarrow[otspam]{}}}

\newcommand{\citet}[2][]{\ifx&#1&\cite{#2}\else\cite[#1]{#2}\fi}
\newcommand{\augmented}{enhanced\xspace}

\newcommand{\augmentedness}{enhancedness\xspace}

\newcommand{\anaugmented}{an enhanced\xspace}

\newtheorem{notation}{Notation}

\crefname{enumi}{}{}
\Crefname{enumi}{}{}
\creflabelformat{enumi}{#2(#1)#3}

\crefformat{section}{\S#2#1#3} 
\crefformat{subsection}{\S#2#1#3}
\crefformat{subsubsection}{\S#2#1#3}
\crefformat{equation}{#2(#1)#3} 
\crefformat{appendix}{\S#2#1#3}

\title{A more general categorical framework for congruence of
  applicative bisimilarity} 

\author{Tom Hirschowitz}{Univ. Savoie Mont Blanc, CNRS, LAMA, \\
  73000, Chambéry, France}{}{https://orcid.org/0000-0002-7220-4067}{}

\author{Ambroise Lafont}{University of Cambridge, United Kingdom}{}{https://orcid.org/0000-0002-9299-641X}{}

\authorrunning{T. Hirschowitz and A. Lafont} 

\Copyright{Tom Hirschowitz and Ambroise Lafont} 


\EventEditors{John Q. Open and Joan R. Access}
\EventNoEds{2}
\EventLongTitle{FSCD 2023}
\EventShortTitle{FSCD 2023}
\EventAcronym{FSCD}
\EventYear{2023}
\EventDate{July 2023}
\EventLocation{Rome, Italy}
\EventLogo{}
\SeriesVolume{42}
\ArticleNo{23}

\begin{document}

\maketitle

\begin{abstract}
  We prove a general congruence result for bisimilarity in
  higher-order languages, which generalises previous
  work~\cite{BHL,HL} to languages specified by a labelled transition
  system in which programs may occur as labels, and which may rely on
  operations on terms other than capture-avoiding substitution. This
  is typically the case for PCF, $λ$-calculus with delimited
  continuations, and early-style bisimilarity in higher-order
  process calculi.
\end{abstract}

\section{Introduction}\label{s:intro}
General congruence results for bisimilarity based on category theory
date back at least to Turi and Plotkin's seminal
paper~\cite{plotkin:turi:bialgebraic}, which covers labelled
transition systems in a categorical version of the Positive GSOS
format~\cite{GSOS}.  The result was then extended to languages with
variable binding and renaming like the
$π$-calculus~\cite{DBLP:conf/lics/FioreT01,DBLP:conf/lics/Staton08}.
More recently, Borthelle et al.~\cite{BHL,HL} managed to deal with a
 wider class of languages, whose operational semantics may rely
not only on renaming, but also on capture-avoiding substitution.

However, their result fails to cover significant languages to which
Howe's method has been adapted, such as (variants of)
PCF~\cite{DBLP:journals/tcs/Gordon99}, $λ$-calculus with delimited
continuations~\cite{DBLP:conf/lfp/DanvyF90,DBLP:conf/fossacs/BiernackiL12},
or (early-style) higher-order process
calculi~\cite{DBLP:books/daglib/0004377,DBLP:conf/concur/LengletS15}.
The reason Borthelle et al.'s framework does not cover such
applications is that they are specified by  labelled transition
systems
\begin{itemize}
\item in which programs may occur as labels, or
\item which rely on operations on terms other than capture-avoiding
  substitution.
\end{itemize}

In this paper, we extend Borthelle et al.'s result to such languages,
which requires a non-trivial extension of the proof method,
essentially abstracting over ideas from
Bernstein~\cite{DBLP:conf/lics/Bernstein98}.

We introduce \alert{algebraic transition systems}, which model
transition systems whose vertices (=states) bear some algebraic
structure, and which may have arbitrary vertices as labels.  For such
transition systems, we define \alert{\enhanced bisimilarity} as an
abstract counterpart to applicative bisimilarity

We introduce \alert{operational signatures}, which allow us to
generate algebraic transition systems of interest, including all
above-mentioned languages. Following initial-algebra
semantics~\cite{goguen1974initial}, an operational signature specifies
algebraic structure and transition rules, and, in applications, the
initial object in the category of models of an operational signature
is the desired syntactic transition system.

Finally, we prove (Theorem~\ref{thm:main}) that, under suitable
conditions, \enhanced bisimilarity in the algebraic transition system
generated by an operational signature is a congruence for the
considered algebraic structure.  We also exhibit
(Theorem~\ref{thm:cellular}) sufficient conditions that are easier to
check in practice.  This covers all above-mentioned applications,
except higher-order process calculi, whose operational signatures fail
to satisfy the required conditions.

\subsection*{Related work}
Beyond Borthelle et al.~\cite{HL}, which was discussed above, the most
closely related work is Goncharov et al.'s~\cite{10.1145/3571215}
bialgebraic framework for higher-order operational semantics. They
upgrade Turi and Plotkin's~\cite{plotkin:turi:bialgebraic} original
presentation of operational semantics as a natural transformation into
a dinatural transformation, which allows them to cover transitions
with programs as labels.  Their main applications are strong variants
of applicative bisimilarity for pure $λ$-calculus (call-by-name and
call-by-value). In its current state, their framework cannot handle
non-deterministic computation, hence in particular weak variants of
bisimilarity.

\subsection*{Plan}
We start in~\cref{ssoverview} with an overview of the development.
In~\cref{serguei}, we then present our running example, which will be
used as the basis of our abstraction process.  We then introduce our
abstract notions of transition systems (\cref{ltss}), and algebraic
transition systems (\cref{atss}), together with bisimilarity and its
\enhanced variant.  Finally, we introduce operational signatures and
state our main results in~\cref{s:howecontexts}, and conclude
in~\cref{sconclu}.

\subsection*{Prerequisites and notations}
We assume some basic knowledge of category theory~\cite{MacLane:cwm},
notably including factorisation systems and monad distributive
laws~\cite{BeckDistlaws}.  Additionally, we rely in places on locally
presentable categories~\cite{Adamek}, but this may be taken as
technical, and ignored on a first reading.  We often conflate natural
numbers $n$ with sets $\{1,…,n\}$.  We denote by $\mathbb{n}$ the
corresponding ordinal viewed as a category, so that, e.g., $𝒞^𝟚$ is
the usual category of morphisms in $𝒞$.  We let $𝐂𝐀𝐓$ denote the
category of locally small categories. Moreover, we denote by $\psh$
the category of (contravariant) presheaves over a given category $ℂ$,
and by $𝐲∶ ℂ → \psh$ the Yoneda embedding.  Furthermore, we recall
that endofunctor algebras differ from monad algebras.  (A monad
algebra structure must be suitably compatible with unit and
multiplication.)  We write $F\alg$ for endofunctor algebras, and
$T\Alg$ for monad algebras (capital `A'!).  Finally, for any
endofunctor $F$ on a sufficiently nice category, e.g., a presheaf
category, we write $F^*$ for the free monad on $F$, which is
furthermore \alert{algebraically free} in the sense that
$F\alg ≅ F^*\Alg$.

\section{Overview}\label{ssoverview}
The development roughly follows~\cite{HL}. We summarise it here,
emphasising the differences.  Our running example throughout is a pure
$λ$-calculus with delimited
continuations~\cite{DBLP:conf/fossacs/BiernackiL12}.

\subsection{Transition systems}
Let us first sketch our notion of transition system, starting from
graphs.  Consider the diagonal functor $Δ∶ 𝐒𝐞𝐭 → 𝐒𝐞𝐭$, defined by
$Δ(X) = X²$.  A graph consists of two sets $E$ and $V$, equipped with
two maps $E → V$, or equivalently a map $E → Δ(V)$.

Hirschowitz and Lafont~\cite{HL} propose a ``typed'' generalisation:
they postulate a category $𝕍𝕋$ of \alert{vertex types}, a category
$𝔼𝕋$ of \alert{edge types}, and two functors $𝐬,𝐭∶ 𝔼𝕋 → 𝕍𝕋$
associating to each edge type the types of its source and target.  A
transition system in their sense consists of a \alert{vertex object}
$V$ in $\psh[𝕍𝕋]$, a \alert{edge object} $E ∈ \psh[𝔼𝕋]$, and a
morphism $E → Δ(V)$, where $Δ∶ \psh[𝕍𝕋] → \psh[𝔼𝕋]$ maps any $V$ to
$Δ(V)(α) = V(𝐬(α)) × V(𝐭(α))$, for all $α ∈ 𝔼𝕋$.
Taking $𝕍𝕋 = 𝔼𝕋 = 1$, one recovers plain graphs.

In this paper, in order to account for labels, we generalise this by
adding a functor $𝐥$ associating to each edge type $α ∈ 𝔼𝕋$ a sequence
$𝐥(α) = (𝐥^α₁,…,𝐥^α_{n_α})$ of vertex types.  A tuple $(𝕍𝕋,𝔼𝕋,𝐬,𝐭,𝐥)$
is called a \alert{Howe context}.  Let us fix one for the rest of this
section.

We modify $Δ$ accordingly, defining it by
$$Δ(V)(α) \qquad = \qquad V(𝐬(α)) \quad × \quad \left ( ∏_{i=1}^{n_α} V(𝐥^αᵢ) \right ) \quad × \quad V(𝐭(α)).$$
A transition system again consists of objects $V ∈ \psh[𝕍𝕋]$ and
$E ∈ \psh[𝔼𝕋]$, together with a morphism $E → Δ(V)$, which means that,
to each edge, we associate a source, a target, and a sequence of
labels of suitable types.  For such transition systems, we define a
generalisation of bisimulation, straightforwardly.

\subsection{Algebraic transition systems and \enhanced bisimilarity}\label{ssalts}
Let us now briefly explain the notion of algebraic structure that we
adopt. Following Fiore et
al.~\cite{fiore:presheaf,DBLP:conf/lics/Fiore08}, Borthelle et
al.~\cite{BHL,HL} use $Σ$-monoids, which are designed to model syntax
with substitution.  In this paper, relying on Hirschowitz and
Lafont~\cite{admissible}, we adopt a different notion of algebraic
structure designed to cover syntax with more general additional operations.
\begin{definition}\label{def:enhanced:syntax}
  An \alert{\enhanced syntax} (on $\psh[𝕍𝕋]$) consists of
\begin{itemize}
\item finitary functors $Σ∶ \psh[𝕍𝕋] → \psh[𝕍𝕋]$ and
  $Γ∶ \psh[𝕍𝕋]² → \psh[𝕍𝕋]$ such that $Γ$ is
  \alert{left-cocontinuous}, i.e., cocontinuous in its first argument,
  equipped with
\item  a distributive
  law $δ∶ T ∘ S → S ∘ T$, where $S = Σ^*$ denotes the monad freely
  generated by $Σ$ and $T = Γ_S^*$ the one generated by $X ↦ Γ(X,S(X))$.
\end{itemize}
\end{definition}
Here, $Σ$ models basic syntax, and $Γ$ models additional operations
like substitution. The fact that $Γ$ is a bifunctor is for
distinguishing a ``main'' occurrence in its arity, which is used below
in the definition of \enhanced bisimilarity.  The distributive law
models commutation of additional operations with basic ones, at the
main occurrence (typically $(M\ N)[σ] = M[σ]\ N[σ]$).

Following initial-algebra semantics~\cite{goguen1974initial}, the main
object of interest here is the initial $Σ$-algebra $S(∅)$, and the
main point is that it automatically possesses $T$-algebra structure,
given by the composite $T (S (∅)) \xto{δ_∅} S (T (∅)) ≅ S (∅)$ (the
initial object is a $T$-algebra by cocontinuity, hence $T(∅) ≅
∅$). This algebra structure in fact makes $S(∅)$ into an initial
algebra for the composite monad $ST$.

Fixing some \enhanced syntax $σ = (Σ,Γ,δ)$, for us, an algebraic
transition system is thus a transition system $E → Δ(V)$, equipped
with $ST$-algebra structure on $V$.  We call such transition systems
\alert{$σ$-algebraic}.

Finally, for any $σ$-algebraic transition system $G = (E,V,∂)$, we
define \alert{\enhanced bisimilarity}, denoted by $∼^σ_G$, as the
greatest bisimulation $R$ which is \alert{\enhanced}, in the sense
that $Γ(R,V) ⊆ R$ -- this is where we use the fact that $Γ$ is a
bifunctor. In concrete instances, as noticed by Borthelle et
al.~\cite{BHL,HL}, \enhanced bisimilarity agrees with applicative
bisimilarity.

The goal is then to prove that, in algebraic transition systems $G$ of
interest, \enhanced bisimilarity $∼^σ_G$ is a \alert{congruence},
i.e., $Σ(∼^σ_G) ⊆ {∼^σ_G}$.

\subsection{Operational signatures}
For this, we restrict attention to algebraic transition systems
generated by a suitable notion of \alert{operational signature}, which
we now describe.  Operational signatures comprise two components, one
for generating \anenhanced syntax, the other for specifying transition
rules.

\begin{definition}\label{def:syntacticsig}
  A \alert{syntactic signature} is an endofunctor $Σ$
  equipped with a sequence
  \begin{equation}
    \hfil T₀ = \id \xto{(Γ₁,d₁)} T₁ \quad … \quad T_{n-1} \xto{(Γₙ,dₙ)} Tₙ\label{eq:syntacticsig}
  \end{equation}
  of \alert{incremental structural laws}~\cite{admissible}.  An
  incremental structural law $T → T'$ consists of a finitary,
  left-cocontinuous bifunctor $Γ∶ \psh[𝕍𝕋]² → \psh[𝕍𝕋]$, together with
  a natural transformation
  $d_{X,Y}∶ Γ (Σ (X),Y) → S (T (Γ (X,S (T (Y))) + X + Y))$, such that
  $T' = T ⊕ Γ_S^*$, where $⊕$ denotes monad coproduct.
\end{definition}
In examples, a natural transformation $d_{X,Y}$ amounts to a
definition by structural recursion, where the first argument of $Γ$
models the decreasing occurrence of the argument, and the second
argument models other occurrences.  Given any syntactic
signature~\cref{eq:syntacticsig}, the given incremental structural
laws induce distributive laws $δᵢ∶ Tᵢ ∘ S → S ∘ Tᵢ$, hence in
particular $δₙ∶ Tₙ ∘ S → S ∘ Tₙ$, and furthermore we have
$Tₙ = (∑ᵢ Γᵢ)_S^*$. Thus, letting $𝐝$ denote the given syntactic
signature, the triple $σ(𝐝) = (Σ,∑ᵢ Γᵢ,δₙ)$ forms \anenhanced syntax.
As a bonus, one can show that algebras for the composite monad $S Tₙ$
are equivalently objects equipped with suitably coherent algebra
structure for $Σ$ and each functor $X ↦ Γᵢ(X,X)$.

The next step is to specify the dynamics of algebraic transition
systems of interest. This is done by introducing \alert{dynamic
  signatures}. Roughly, a dynamic signature over \anenhanced syntax
$σ$ is an endofunctor on $σ$-algebraic transition systems, which is
required to preserve the vertex object and satisfy a suitable
``structuralness'' condition inspired by structural operational
semantics~\cite{PlotkinSOS}. Intuitively, a dynamic signature $Σ₁$ is
a family of transition rules, and structuralness demands that, in each
transition rule, the source of the conclusion has depth at most one.

Pursuing the analogy, $σ$-algebraic transition systems satisfying the
rules are a special kind of $Σ₁$-algebras which we call
\alert{vertical}. Verticality means that the algebra structure is
trivial on vertices: this enforces that satisfying the rules is only
about edges, not vertices. 

Finally, an \alert{operational signature} consists of a syntactic
signature $𝐝$, and a dynamic signature $Σ₁$ on $σ(𝐝)$.  The real
object of interest is here the initial vertical $Σ₁$-algebra, say
$𝐙 = (E_𝐙,V_𝐙,∂_𝐙)$, which in applications is the desired syntactic
transition system.

\subsection{Congruence of \enhanced bisimilarity}
Our goal is then to prove that, under suitable hypotheses, \enhanced
bisimilarity $∼^{σ(𝐝)}_𝐙$ in the initial vertical $Σ₁$-algebra is a
congruence.  For this, abstracting over
Bernstein's~\cite{DBLP:conf/lics/Bernstein98} proof, we start by
defining \alert{flexible bisimulation}, a variant of Sangiorgi's
BA-bisimulation~\cite{DBLP:conf/fsen/SangiorgiKS07}.  Flexible
bisimulation is like plain bisimulation: given related elements $e$
and $e'$, any transition from $e$ should be matched by some transition
from $e'$. The difference is that, instead of having the same label,
the matching transition should exist for any related label.  Defining
\alert{functional flexible bisimulations} to be morphisms of algebraic
transition systems whose graph is a flexible bisimulation, our main
result (Theorem~\ref{thm:main}) states that if the dynamic
signature $Σ₁$ preserves functional flexible bisimulations, then
$∼^{σ(𝐝)}_𝐙$ is a congruence.

Finally, preservation of functional flexible bisimulations is quite an
abstract condition, so we set out to design a more concrete criterion
for making the result easier to apply.  In fact, if the considered
dynamic signature $Σ₁$ is
\alert{familial}~\cite{Diers1978Spectres,DBLP:journals/mscs/CarboniJ95,WeberPhD,garner:hal-01246365},
then preservation of functional flexible bisimulations becomes quite
tractable, as we now explain.  Following Joyal et
al.~\cite{DBLP:conf/lics/JoyalNW93}, we first characterise functional
flexible bisimulations as the right class of a weak factorisation
system~\cite{Hovey,riehl} -- we call the left class
\alert{cofibrations}.  Furthermore, when the dynamic signature is
familial, a transition rule with conclusion of type any $α$, is
intutively an element of $Σ₁(1)(α)$, and we extract for each rule two
algebraic transition systems $A$ and $B$, and a morphism $ϕ∶ A → B$,
such that, intuitively, $A$ describes the metavariables occurring in
the source and label of the conclusion, $B$ describes all
metavariables in the rule, including transition premises, and $ϕ$
embeds the former into the latter. We call $ϕ$ the \alert{border
  arity} of the rule.  The main point is then that a familial $Σ₁$
preserves functional flexible bisimulations iff all border arities are
cofibrations (\cref{thm:cellular}).  How is this any more concrete?
Well, cofibrations are well-known to be closed under composition and
cobase change, so in order to check preservation of functional
flexible bisimulations, it suffices to reconstruct the border arity of
each rule from generating cofibrations, by composition and cobase
change.  This reconstruction process is close in spirit to usual
acyclicity
criteria~\cite{DBLP:journals/iandc/Howe96,DBLP:conf/lics/Bernstein98}.
\begin{example}\label{ex:graphborderarity}
  Taking algebraic transition systems to be just plain graphs, for a
  rule like $\inferrule{a → b \\ b → c}{a → c}$, $A$ would be the
  one-vertex graph, $B$ would consist of two composable edges
  $x → y → z$, and $ϕ$ would pick $x$.  To check that it is a
  cofibration, we reconstruct it as the bottom composite in
\qquad    \Diag|baseline=(m-1-2.base)|{%
      \pbk{m-2-2}{m-2-3}{m-1-3} %
    }{%
      \& {[0]} \& {[1]} \\
      {\llap{$A = {}$} [0]} \& {[1]} \& B.
    }{%
      (m-1-2) edge[labela={s}] (m-1-3) %
      edge[labell={}] (m-2-2) %
      (m-2-2) edge[labela={}] (m-2-3) %
      (m-1-3) edge[labelr={}] (m-2-3) %
      (m-2-1) edge[labela={s}] (m-2-2) %
    }
\end{example}
As an application, we recover congruence of applicative bisimilarity in
the considered $λ$-calculus with delimited
continuations~\cite{DBLP:conf/fossacs/BiernackiL12}.

\section{A concrete example}\label{serguei}
As a concrete example result that we want to abstract over, let us
recall the case of $λ$-calculus with delimited continuations.  We
present it in a non-standard way in order for it to fit the abstract
framework.  Indeed, the framework is based on \alert{structural}
operational semantics~\cite{PlotkinSOS}, in the sense that, in each
transition rule, the source of the conclusion has depth at most
one. Following~\cite{BHL,HL}, we also present the definition of the
open extension of applicative bisimilarity to make it compatible with
the abstract developments to come.

The syntax, presented in the usual, informal way, is as below left,
\begin{align}
  \mbox{Values} ∋ v & ::=  x  ｜  λx.e                         &          ▫[e] & = e  \label{eq:wbox} \\
  \mbox{Programs} ∋ e & ::=  v  ｜  e₁\ e₂  ｜  𝒮x.e  ｜  ⟨e⟩  &     (v\ E)[e] & =  v\ E[e] \label{eq:appv} \\
  \mbox{Evaluation contexts} ∋ E & ::=  ▫  ｜ E\ e  ｜  v\ E          &      (E\ e')[e] & =  E[e]\ e'. \label{eq:appe}
\end{align}
where $x$ binds in $e$, in both $λx.e$ and $𝒮x.e$.
Capture-avoiding substitution and context application are defined as
usual. E.g., context application is defined as above right.  The
dynamics are governed by the rules in~\cref{fig:trans:shiftreset}.
\begin{figure}[htp]
\begin{mathpar}
  \inferrule{e₁ \xto{v} e₂}{e₁\ v \xto{τ} e₂}~(β') \and %
  \inferrule{ }{λx.e \xto{v} e[x↦v]} \and
  \inferrule{e₁ \xto{τ} e'₁}{e₁\ e₂ \xto{τ} e'₁\ e₂}
  \and
  \inferrule{e₂ \xto{τ} e'₂}{v\ e₂ \xto{τ} v\ e'₂} \and %
  \inferrule{ }{⟨v⟩ \xto{τ} v} \and
  \inferrule{e \xto{τ} e'}{⟨e⟩ \xto{τ} ⟨e'⟩} \and
  \inferrule{e \xto{▫} e'}{⟨e⟩ \xto{τ} e'} \and
  \inferrule{e₁ \xto{E[▫\ e₂]} e₃}{e₁\ e₂ \xto{E} e₃} \and
  \inferrule[(SA)]{e₁ \xto{E[v\ ▫]} e₂}{v\ e₁ \xto{E} e₂} \and
  \inferrule{ }{𝒮k.e \xto{E} ⟨e [ k ↦ λx.⟨E[x]⟩ ]⟩} \and
  \inferrule{ }{e \xto{τ} e} \and
  \inferrule{e₁ \xto{τ} e₂ \xto{α} e₃}{e₁ \xto{α} e₃} \and
  \inferrule{e₁ \xto{α} e₂ \xto{τ} e₃}{e₁ \xto{α} e₃} \and %
  \mbox{Deriving $β$:}\quad \inferrule{
    \inferrule*{
    }{λx.e \xto{v} e[x↦v]         }
  }{(λx.e)\ v \xto{τ} e[x↦v]}~(β')   
\end{mathpar}     
\caption{Transition rules}
\label{fig:trans:shiftreset}
\end{figure}
There are three kinds of transitions, of types $e \xto{τ} e'$,
$e \xto{v} e'$, $e \xto{E} e'$, where all expressions are closed.  The
first four rules deal with functions. The first two rules suffice to
make (β) derivable, as shown in~\cref{fig:trans:shiftreset}. The next
two rules are the usual context rules.  The last three rules, where
$α$ ranges over all labels, enforce that we work with weak
bisimulation: we close transitions under composition with silent
transitions.  The remaining rules describe the dynamics of $𝒮x.e$ and
$⟨e⟩$, which are respectively called \alert{shift} and \alert{reset}.
The first two of them enforce that silent computation occurs normally
inside any reset, and if it succeeds, i.e., if it results in a value,
then the reset disappears.  The next rules describe how shift captures
the ambient context up to the enclosing reset, say $E$, and
substitutes its reification $λk.⟨E[k]⟩$ as a value for the bound
variable, placing a new reset around the result.

Bisimulation is then as expected:
\begin{definition}
  A binary relation $R$ between closed programs is a
  \alert{simulation} iff for all $e \mathrel{R} e'$ and transitions
  $e \xto{α} e₁$, there exists a transition $e' \xto{α} e'₁$ such
  that $e₁ \mathrel{R} e'₁$.  A \alert{bisimulation} is a
  simulation whose converse relation also is a simulation.
\end{definition}
\begin{definition}[{\cite{BHL,HL}}]
  A relation $R$ on potentially open expressions is \alert{\enhanced}
  iff it is closed under substitution, context composition, and
  context application, i.e., $a \mathrel{R} a'$ entails
  $a[σ] \mathrel{R} a'[σ]$ for all substitutions $σ$,
  $E \mathrel{R} E'$ entails $E[e] \mathrel{R} E'[e]$ and
  $E[E''] \mathrel{R} E'[E'']$, for all $e$ and $E''$.

  An \alert{\enhanced bisimulation} is \anenhanced relation $R$ whose
  restriction to closed programs is a bisimulation.
\end{definition}
\begin{proposition}\label{prop:enhancedbisim}
  There is a largest \enhanced bisimulation, called
  \alert{applicative bisimilarity}.
\end{proposition}

The result that we want to abstract over is:
\begin{theorem}[generalised variant of {\cite[Theorem 1]{DBLP:conf/fossacs/BiernackiL12}}]\label{thm:serguei}
  Applicative bisimilarity is a \alert{congruence}, in the sense that
  it is preserved by all constructions of the language.
\end{theorem}
\begin{remark}
  It is not entirely trivial that this agrees with Biernacki and
  Lenglet's presentation.  In fact, their transition system only
  differs in that they replace rule $(β')$ with the standard rule
  $(β)$.  We have already seen that $(β)$ is derivable from $(β')$,
  and conversely $(β')$ is admissible in their transition system.
  Indeed, suppose given any transition $e₁ \xto{v} e₂$.  By an easy
  induction, there exist transitions
  $e₁ \xto{τ} λx.e₃ \xto{v} e₃[x↦v] \xto{τ} e₂$.  Hence, grouping
  saturation rules, we derive $(β)$ as follows.
  \begin{mathpar}
    \inferrule*{%
      \inferrule*{
      e₁ \xto{τ} λx.e₃
      }{
      e₁\ v \xto{τ} (λx.e₃)\ v         
        } \\
      (λx.e₃)\ v \xto{τ} e₃[x↦v] \\
      e₃[x↦v] \xto{τ} e₂
    }{%
      e₁\ v \xto{τ} e₂
      }
    \end{mathpar}
\end{remark}
Our problem is that this result is not an instance of Borthelle et
al.'s~\cite[Theorem~6.15]{HL}, because
the dynamics rely on two features that are not handled:
\begin{alphaenumerate}
\item \label{item:contapp} operations on terms, context application
  and composition, which differ from substitution, 
\item \label{item:contlab} and contexts and values occurring as
  labels.
\end{alphaenumerate}
For~\cref{item:contapp}, context application and composition might be
encodable in Borthelle et al.'s setting, perhaps by resorting to the
skew monoidal variant~\cite{BHL}. But this is quite artificial, and
requires extra work that should not be necessary.
For~\cref{item:contlab}, it appears to be a hard obstruction.

\section{Transition systems in the abstract}\label{ltss}
In this section, we start to abstract over the development
of~\cref{serguei}, by introducing a notion of labelled transition
system, together with its associated notion of bisimilarity.

\subsection{Howe contexts}
Let us start by formally introducing Howe contexts, as sketched
in~\cref{ssoverview}.
\begin{definition}\label{def:howcont}
  A \alert{Howe context} consists of
  \begin{itemize}
  \item a small category $𝕍𝕋$ of \alert{state types},
  \item a small category $𝔼𝕋$ of \alert{transition types},
  \item \alert{source} and \alert{target} functors
      $\source,\but∶ 𝔼𝕋 → 𝕍𝕋$, and
    \item a \alert{label} functor $\labels∶ 𝔼𝕋 → \psh[𝕍𝕋]$, such that
      each $\labels(c)$ is a finite coproduct of representables.
\end{itemize}
\end{definition}

  \begin{example}
    For plain graphs, we would take:
    \begin{itemize}
    \item $𝕍𝕋$ to be the terminal category, because there is just one
      kind of vertex,
    \item $𝔼𝕋$ to also be the terminal category, because there is just one
      kind of edge, 
    \item the source and target functors both are the unique functor
      $1 → 1$, and
    \item the label functor to map the unique object to the empty
      coproduct, i.e., $∅$.
    \end{itemize}
  \end{example}

  \begin{example}\label{ex:src}
    For modelling the transition system of~\cref{serguei}, we need a
    presheaf on $𝕍𝕋$ to be equivalent to a triple of functors
    $V_𝐩,V_𝐯,V_𝐜∶ 𝔽 → 𝐒𝐞𝐭$, where $𝔽$ denotes a skeleton of the
    category of finite sets, e.g., finite ordinals and all maps
    between them, equipped with a natural transformation
    $ι∶ V_𝐯 → V_𝐩$, or otherwise said to a functor $𝔽 → 𝐒𝐞𝐭^{1 + 𝟚}$.
    We think of $V_𝐩(n)$, $V_𝐯(n)$, and $V_𝐜(n)$ as sets of programs,
    values, and contexts with $n$ free variables, respectively.  For
    making this into a presheaf category, let us first observe that
    such tuples $(V_𝐩,V_𝐯,V_𝐜,ι)$ are precisely the objects of the
    oplax limit of the functor
    $Δ_{𝐲 in₁}∶ \psh[\op{𝔽}+\op{𝔽}] → \psh[\op{𝔽}]$ mapping any
    copairing $[V_𝐩,V_𝐜]$ to $V_𝐩$.  But, as we now recall, oplax
    limits of this form are equivalent to presheaf categories.
    \begin{definition}\label{def:collage}
      For any small categories $𝕏$ and $𝕐$, and functor
      $F∶ 𝕏 → \psh[𝕐]$, the \alert{collage} of $F$, denoted by
      $𝕐[𝕏]_F$, or merely $𝕐[𝕏]$ when $F$ is clear from context, has
      as objects the disjoint union of those of $𝕏$ and $𝕐$, and
      morphisms defined by cases as follows.
      \begin{center}
        $\begin{array}{rcl}
           𝕏[𝕐](x, x')  & = &  𝕏(x, x') \\
           𝕏[𝕐](y, y')  & = &  𝕐(y, y')
         \end{array}$ \hfil
         $\begin{array}{rcl}
            𝕏[𝕐](y, x)  & = &  F(x)(y) \\
            𝕏[𝕐](x, y)  & = &  ∅
          \end{array}$
        \end{center}
        Composition is defined as in $𝕏$ and $𝕐$ in both left-hand cases, and
        otherwise by action of $F$.
      \end{definition}
\begin{proposition}[{\cite[Lemma~4.9]{DBLP:journals/mscs/CarboniJ95}}]
  \label{prop:coll}
  For any small categories $𝕏$ and $𝕐$, and functor $F∶ 𝕏 → \psh[𝕐]$,
  letting $Δ_F(Y)(x) = \psh[𝕐](F(x),Y)$ denote the induced
  \alert{nerve} functor $\psh[𝕐] → \psh[𝕏]$, the oplax limit
  $\psh[𝕏]/Δ_F$ is equivalent to the category $\psh[𝕐[𝕏]]$ of
  presheaves on the collage of $F$.
\end{proposition}
Now, the above functor $Δ_{𝐲 in₁}$ is indeed the nerve of
$\op{𝔽} \xto{in₁} \op{𝔽}+\op{𝔽} \xto{𝐲} \psh[\op{𝔽}+\op{𝔽}]$
since we have
$Δ_{𝐲 in₁}[V_𝐩,V_𝐜](n) = V_𝐩(n) = [V_𝐩,V_𝐜](in₁(n)) = 
\psh[\op{𝔽}+\op{𝔽}](𝐲 (in₁ (n)),[V_𝐩,V_𝐜])$. We obtain:
\begin{corollary}
  Letting $𝕍𝕋 = {(\op{𝔽}+\op{𝔽})[\op{𝔽}]_{𝐲 in₁}}$, we have
  $[𝔽,𝐒𝐞𝐭^{1 + 𝟚}] ≃ \psh[𝕍𝕋]$.
\end{corollary}

\begin{notation}
  We denote objects $in₁ n$, $in₂ n$, and $in₃ n$ of $𝕍𝕋$ by $n_𝐯$,
  $n_𝐩$, $n_𝐜$, respectively, for values, programs, and contexts. For
  any $V ∈ \psh[𝕍𝕋]$, we denote the corresponding functors $𝔽 → 𝐒𝐞𝐭$
  by $V_𝐯$, $V_𝐩$, and $V_𝐜$, so that, e.g., $V(n_𝐯) = V_𝐯(n)$.
  \end{notation}

  Let us now define $𝔼𝕋 = 3 = \{ [τ], [𝐯], [𝐜] \} $, where $[α]$
  indicates a label of type $α$. Accordingly, writing $c∶ a \xto{L} b$
  for $\source(c) = a$, $\labels(c) = L$, and $\but(c) = b$, and
  respectively interpreting $τ$, $𝐯$, and $𝐜$ as $∅$, $𝐲_𝐯$, and
  $𝐲_𝐜$, we put: $[α]∶ 0_𝐩 \xto{α} 0_𝐩$, for all $α ∈ \{τ,𝐯,𝐜\}$.
  \end{example}

\subsection{Generalised transition systems}  
Let us now introduce transition systems.  Let us fix a Howe context
$ℍ = (𝕍𝕋,𝔼𝕋,\source,\but,\labels)$ for the whole subsection, and start
by relating both categories $\psh[𝕍𝕋]$ and $\psh[𝔼𝕋]$.
\begin{definition}
We define four functors
$
\psh[𝕍𝕋] → \psh[𝔼𝕋]$
as follows,  for all  $V ∈ \psh[𝕍𝕋]$ and $α ∈ 𝔼𝕋$.
\begin{center}
$\begin{array}[t]{rcll}
  Δ_\source(V)(α) & = & V(\source(α)) \\
  Δ_\but(V)(α) & = & V(\but(α))
\end{array}$ \hfil
$\begin{array}[t]{rcll}
  Δ_\labels(V)(α) & = & \psh[𝕍𝕋](\labels(α),V) \\
  Δ_ℍ(V) & = & Δ_\source(V) × Δ_\labels(V) × Δ_\but(V)\rlap{,}
\end{array}$
\end{center}
\end{definition}
\begin{notation}\label{not:Delta}
  We often abbreviate $Δ_ℍ$ to $Δ$ when $ℍ$ is clear from context.  We
  also use juxtaposition of indices to denote product of the
  corresponding functors, e.g., $Δ_{𝐬,𝐥} ≔ Δ_𝐬 × Δ_𝐥$.
\end{notation}

\begin{definition}
  An \alert{$ℍ$-transition system} $G$ consists of a \alert{vertex}
  presheaf $V_G ∈ \psh[𝕍𝕋]$, an \alert{edge} presheaf
  $E_G ∈ \psh[𝔼𝕋]$, and a \alert{border} natural transformation
  $∂_G∶ E_G → Δ(V_G)$.
\end{definition}

\begin{remark}
  Letting $\labels(α) = ∑_{i ∈ n_α} 𝐲_{𝐥^αᵢ}$, we have
  $Δ_\labels(V)(α) = [∑_{i ∈ n_α} 𝐲_{𝐥^αᵢ},V] ≅ ∏_{i ∈ n_α} V(𝐥^αᵢ)$
  for any $α ∈ 𝔼𝕋$ and $V ∈ \psh[𝕍𝕋]$.  The border natural
  transformation thus has type
  \begin{center}
    $E(α) → V(\source(α)) × (∏_{i ∈ n_α} V(𝐥^αᵢ)) × V(\but(α)).$
  \end{center}
\end{remark}
\begin{example}
  Let us unfold the definition for the Howe context of
  \cref{ex:src}: a transition system
  consists of presheaves $V ∈ \psh[𝕍𝕋]$ and $E ∈ \psh[𝔼𝕋]$,
  equipped with maps
  \begin{mathpar}
    E[τ] → V_𝐩(0)² \and
    E[𝐯] → V_𝐩(0) × V_𝐯(0) × V_𝐩(0) \and
    E[𝐜] → V_𝐩(0) × V_𝐜(0) × V_𝐩(0).
\end{mathpar}
\end{example}

We now equip $ℍ$-transition systems with morphisms:
\begin{proposition}
  \label{prop:h-trans-comma}
  $ℍ$-transition systems
  are precisely the objects of the oplax limit category $\psh[𝔼𝕋]/Δ$ of
  the functor $\psh[𝕍𝕋] \xto{Δ} \psh[𝔼𝕋]$ in $𝐂𝐀𝐓$, or equivalently the
  comma category $\id_{\psh[𝔼𝕋]} ↓ Δ$.
\end{proposition}
\begin{proof}
  An object of the oplax limit is by definition a triple $(E,V,{∂∶ E → Δ(V)})$.
\end{proof}

\begin{definition}\label{def:HTrans}
 Let $ℍ\Trans = \psh[𝔼𝕋]/Δ_ℍ$.
\end{definition}

\subsection{Bisimulation and bisimilarity}\label{ss:bisimex}
We now want to define bisimulation and bisimilarity, for any fixed
Howe context $ℍ = (𝕍𝕋,𝔼𝕋,𝐬,𝐭,𝐥)$.  Let us start with the notion of
simulation.
\begin{notation}
A \alert{span} is a pair of morphisms with the same source.
In a category with binary products, we often write spans
$X ← R → Y$ as their pairings $R → X×Y$.
The \alert{converse} of a span $⟨f,g⟩∶ R → X×Y$ is the composite
$⟨g,f⟩∶ R → X×Y$.

In a presheaf category $\psh$, for any span $j∶ R → X×Y$, object
$c ∈ ℂ$, and element $r ∈ R(c)$, we write $r∶ x \mathrel{R} y$ when
$j_c(r) = (x,y)$.  We call $r$ a \alert{witness} that $x$ and $y$ are
related by $R$.

Finally, in any $ℍ$-transition system $G$, for any transition type $α$
with $\labels(α) ≅ ∑_{i ∈ n_α} 𝐲_{𝐥^αᵢ}$, we write
$e∶ x \xto{α(l₁,…,l_{n_α})} y$ to mean that
$e ∈ E_G(α)$ and
$∂_G(e) = (x,(l₁,…,l_{n_α}),y)$.
\end{notation}
\begin{definition}\label{def:simulation}
  For any $ℍ$-transition system $G = (V,E,{∂∶ E → ΔV})$, a given span
  $j∶ R → V²$ is a \alert{simulation} when, for any transition
  $e∶ x \xto{α(l₁,…,l_{n_α})} x'$ and witness $r∶ x \mathrel{R} y$,
  there exists a transition $f∶ y \xto{α(l₁,…,l_{n_α})} y'$ and a witness
  $r'∶ x' \mathrel{R} y'$, as in
\begin{equation}
\hfil  \diag{%
    x \& R(𝐬(α)) \& y \\
    x' \& R(𝐭(α)) \&  y'. %
  }{%
    (m-1-1) 
    edge[labell={e∶ α(l₁,…,l_{n_α})}] (m-2-1) %
    (m-1-3) edge[labelr={f∶ α(l₁,…,l_{n_α})}] (m-2-3) %
  }\label{eq:compattuple}
\end{equation}
  A span is a \alert{bisimulation} when it is a simulation and so is its
  converse.
  A \alert{bisimulation relation} is a bisimulation which is also a relation,
  i.e., a mono $R ↪ V²$.
\end{definition}

\begin{proposition}\label{prop:bisimilarity}
  The full subcategory $𝐁𝐢𝐬𝐢𝐦(G)$ of $\psh[𝕍𝕋]/V²$ spanning
  bisimulations admits a terminal object, which we call
  \alert{bisimilarity} and denote by $∼_G$.
\end{proposition}
\begin{proof}
  Bisimulation relations are stable under unions, so that
  a terminal object is given by the union of them all.
\end{proof}

\section{Algebraic transition systems}\label{atss}
In this section, we explain \enhanced syntax, algebraic transition
systems, and \enhanced bisimulation in a bit more detail than
in~\cref{ssalts}.  The notion of \enhanced syntax has already been
introduced (\cref{def:enhanced:syntax}), and we fix a Howe context
$ℍ = (𝕍𝕋,𝔼𝕋,𝐬,𝐭,\labels)$ and \anenhanced syntax
$σ=(Σ,Γ,{δ∶ TS → ST})$, where, we recall, $S=Σ^*$ and $T=Γ_S^*$.
\subsection{\Enhanced syntax}
\begin{definition}
  We call $ST$-algebras \alert{$σ$-algebras} for short, and let
  $σ\Alg = ST\Alg$.
\end{definition}

\begin{proposition}
  The initial $Σ$-algebra $S∅$ is automatically a $T$-algebra, with
  structure map $TS∅ \xto{δ_∅} ST∅ \xto{≅} S∅$.
\end{proposition}
\begin{proof}
  By cocontinuity, $∅$ is a $Γ_S$-algebra: we have $Γ(∅,S(∅)) ≅ ∅$. It
  is thus an initial $Γ_S$-algebra, hence an initial $T$-algebra
  since $Γ_S\alg ≅ T\Alg$.
\end{proof}

\begin{example}\label{ex:srcalg}
  Following up on \cref{ex:src}, the syntax and additional operations
  of \cref{serguei} may be presented by an incremental structural law
  on $\psh[𝕍𝕋]$, as follows.  First, basic operations are specified by
  the endofunctor $Σ₀$ defined as follows (recalling original notation
  on the right).

\noindent$\begin{array}{r@{\ =\ }lr@{\ ::=\ }l}
Σ₀(X)_𝐯(n) &   n  +  X_𝐩(n + 1) & v &  x  ｜  λx.e    \\
     Σ₀(X)_𝐩(n) &   Σ₀(X)_𝐯(n) + X_𝐯(n)  +  X_𝐩(n)²  +  X_𝐩(n + 1)  +  X_𝐩(n) &  e &  v  ｜  e₁\ e₂  ｜  𝒮x.e  ｜  ⟨e⟩  \\
     Σ₀(X)_𝐜(n) &   1  +  X_𝐯(n) × X_𝐜(n)  +  X_𝐜(n) × X_𝐩(n)
                & E &  ▫  ｜ E\ e  ｜  v\ E 
   \end{array}$

   \noindent
   We then want to define the arity of additional operations, namely
   substitution, context application, and context composition.  Since
   these three additional operations are independent, we may specify
   them at once by the bifunctor $Γ∶ \psh[𝕍𝕋]² → \psh[𝕍𝕋]$ defined as
   follows.

   \noindent$\begin{array}{rclrcl}
   Γ(X,Y)_𝐯(n) & = &   \textstyle ∑_{m ∈ ℕ} X_𝐯(m) × Y_𝐯(n)ᵐ &
                                          v & {+}\Coloneqq & v[σ]    \\
   Γ(X,Y)_𝐩(n) & = &   \textstyle ∑_{m ∈ ℕ} X_𝐩(m) × Y_𝐯(n)ᵐ  +  X_𝐜(n) × Y_𝐩(n) &
                                                              e & {+}\Coloneqq & e[σ] ｜ E[e]    \\   
   Γ(X,Y)_𝐜(n) & = &   X_𝐜(n) × Y_𝐜(n) &
    E & {+}\Coloneqq & E[E']
 \end{array}$

 \noindent That the actual definition of additional operations induces
 a distributive law of $Γ_S^*$ over $Σ^*$ is harder to see, and will
 follow from the theory of syntactic signatures below
 (\cref{ex:srcii}).
\end{example}

\subsection{Algebraic transition systems}
Let us now introduce algebraic transition systems.
\begin{definition}
  A \alert{$σ$-transition system} is an $ℍ$-transition systems
  equipped with $σ$-algebra structure on its vertex object.  A
  $σ$-transition system morphism is a morphism of $ℍ$-transition
  systems whose vertex component is a $σ$-algebra morphism.  Let
  $σ\Trans$ denote the category of $σ$-transition systems and
  morphisms between them.
\end{definition}
\begin{proposition}\label{prop:L}
  The forgetful functor $𝒰$ has a left adjoint, say
  $ℒ∶ ℍ\Trans → σ\Trans$.
\end{proposition}
\begin{proof}
  The left adjoint maps any $∂∶ E → Δ(V)$ to
  $E \xto{∂} Δ(V) \xto{Δ(η^{ST}_V)} Δ(S (T (V)))$.
\end{proof}

We conclude this section by defining the notion of congruence.
\begin{definition}
  For any $σ$-transition system $G = (V,E,∂)$, a \alert{congruence} is
  a span $R → V²$ for which there exists a morphism $Σ(R) → R$ making
  the first diagram of~\cref{fig:congenhance} commute.
  \begin{figure}[htp]
    \centering
    \diag{%
      Σ(R) \& \& R \\
      Σ(V²) \& Σ(V)² \& V² %
    }{%
      (m-1-1) edge[dashed,labela={}] (m-1-3) %
      edge[labell={}] (m-2-1) %
      (m-2-1) edge[loinb={⟨Σ (π₁),Σ (π₂)⟩}] (m-2-2) %
      (m-2-2) edge[labelb={}] (m-2-3) %
      (m-1-3) edge[labelr={}] (m-2-3) %
    }
    \hfil
        \diag{%
      Γ(R,V) \& \&  R \\
      Γ(V²,V) \& Γ(V,V)² \& V² }{%
      (m-1-1) edge[dashed,labela={}] (m-1-3) %
      edge[labell={},shorten >=2pt] (m-2-1) %
      (m-2-1) edge[loinb={⟨Γ (π₁,V),Γ (π₂,V)⟩}] (m-2-2) %
      (m-2-2) edge[labelb={}] (m-2-3) %
      (m-1-3) edge[labelr={}] (m-2-3) %
    }
    \caption{Congruence and enhancement}
    \label{fig:congenhance}
  \end{figure}
\end{definition}

\subsection{\Enhanced bisimilarity}
\label{s:enhancedbisim}

\begin{definition}
  \label{def:enhanced-span}
  For any $σ$-algebra $V$, a span $p∶ R → V²$ is \alert{\enhanced}
  when there exists a morphism $Γ(R,V) → R$ making the second diagram
  of~\cref{fig:congenhance} commute.
\end{definition}

\begin{definition}
  For any $σ$-transition system $G$, let $𝐁𝐢𝐬𝐢𝐦^σ(G)$ denote the full
  subcategory of $𝐁𝐢𝐬𝐢𝐦(G)$ on \enhanced spans.  We call such spans
  \alert{\enhanced bisimulations}.
\end{definition}

\begin{proposition}
  For any $σ$-transition system $G$, $𝐁𝐢𝐬𝐢𝐦^σ(G)$ admits a terminal
  object, which we call \alert{\enhanced} bisimilarity and denote by
  $∼^σ_G$.
\end{proposition}
\begin{proof}
  Similar to \cref{prop:bisimilarity}, using left-cocontinuity of $Γ$.
\end{proof}

\begin{example}\label{ex:srcenhancedbisim}
  In the setting of~\cref{ex:srcalg}, \enhanced bisimilarity is
  applicative bisimilarity.
\end{example}

\section{Signatures for operational semantics}\label{s:howecontexts}

\subsection{Syntactic signatures for \enhanced syntax}
Syntactic signatures have already been introduced
in~\cref{def:syntacticsig}.

\begin{example}\label{ex:srcii}
  Following up on \cref{ex:srcalg}, the syntax and additional
  operations of \cref{serguei} may be presented as an incremental
  structural law $d_{X,Y}∶ Γ_Y(Σ(X)) → S (Γ_{S(Y)}(X) + X + Y)$
  (taking $T₁ = \id$) on $\psh[𝕍𝕋]$, as follows.  For context
  application, \crefrange{eq:wbox}{eq:appe} may be interpreted as the
  component $Σ(X)_𝐜(n) × Y_𝐩(n) → S (Γ_{S(Y)}(X) + X + Y)_𝐩(n)$,
  namely we take them to mean
 $$\begin{array}{rcl}
   (in₁(⋆),y) & ↦ & in'₃(y) \\
   (in₂(v,E),y) & ↦ & ι(in'₂(v))\ in'₁(E,y) \\
   (in₃(E,e),y) & ↦ & in'₁(E,y)\ in'₂(e)\rlap{,}
   \end{array}$$
   where $in'ᵢ = η^S ∘ inᵢ$.
   For context composition,
   we define the component
   at $𝐜$ (for any $n$), by  
   the exact same formulas, only with $y ∈ Y_𝐜(n)$.
   Substitution is defined similarly~\cite{fiore:presheaf,BHL,HL}.
\end{example}

\begin{proposition}\label{prop:sigmad}
  For any syntactic signature $𝐝 = (Σ,(Γᵢ,dᵢ)_{i ∈ n})$ as
  in~\cref{eq:syntacticsig}, the given incremental structural laws
  induce distributive laws $δᵢ∶ Tᵢ ∘ S → S ∘ Tᵢ$, hence in particular
  $δₙ∶ Tₙ ∘ S → S ∘ Tₙ$, and furthermore we have $Tₙ = (∑ᵢ
  Γᵢ)_S^*$. Thus, the triple $σ(𝐝) = (Σ,∑ᵢ Γᵢ,δₙ)$ forms \anenhanced
  syntax.
\end{proposition}
\begin{proof}
  By \cite[Theorem~4.2]{admissible}, each incremental structural law
  $dᵢ$ induces a distributive law of $(T_{i-1} ⊕ Γ_S^*)$ over $S$,
  i.e., of $Tᵢ$ over $S$ by definition, using $(F+G)^* ≅ F^* ⊕ G^*$.
\end{proof}

Let us conclude this subsection by giving an explicit description of
the algebras of the composite monad $STₙ$ generated by a
syntactic signature.
\begin{definition}\label{def:enhanced}
  Consider any syntactic signature $𝐝 = (Σ,(Γᵢ,dᵢ)_{i ∈ n})$.  For
  $i ∈ n$, an \alert{\enhanced algebra} is an object equipped with
  algebra structures $a∶ ΣX → X$, $b₁∶ Γ₁(X,X) → X$, …,
  $bₙ∶ Γₙ(X,X) → X$ such that for all $i ∈ n$ the following diagram
  commutes,
  \begin{center}
\begin{tikzcd}[ampersand replacement=\&]
	{\Gamma_{i}(\Sigma X,X)} \& {ST_{i}(\Gamma_{i}(X,ST_{i}X)+X+X)} \&\& {ST_i(\Gamma_i(X,X)+X)} \& {ST_iX} \\
	{\Gamma_i(X,X)} \&\&\&\& X
	\arrow["{(d_{i})_{X,X}}", from=1-1, to=1-2]
	\arrow[loina={ST_i(\Gamma_i(X, \bar{a} \circ S\bar{a}_i)+[X,X])}, from=1-2, to=1-4]
	\arrow["{ST_i[b_i,X]}", from=1-4, to=1-5]
	\arrow["{\Gamma_i(a,X)}"', from=1-1, to=2-1]
	\arrow["{b_i}"', from=2-1, to=2-5]
	\arrow["{\bar{a} \circ S\bar{a}_i}", from=1-5, to=2-5]
\end{tikzcd}    
  \end{center}
  where $\bar{a}ᵢ∶ Tᵢ X → X$ and $\bar{a}∶SX → X$ denote the algebra
  structures induced by $(bⱼ)_{j < i}$, and $a$.  Let $𝐝\Alg$ denote
  the full subcategory of $(Σ+∑_{i ∈ n} Γᵢ)\alg$ spanned by \enhanced
  algebras.
\end{definition}

\begin{proposition}
  Let $𝐝 = (Σ,(Γᵢ,dᵢ)_{i ∈ n})$ denote any syntactic signature.  The
  forgetful functor $σ(𝐝)\Alg → (Σ+∑_{i ∈ n} Γᵢ)\alg$ lifts to
  $𝐝\Alg$, and the lifting is an isomorphism.  In short, we have
  $σ(𝐝)\Alg ≅ 𝐝\Alg$ over $\psh$.
\end{proposition}
\begin{proof}
  By induction on $n$ and~\cite[Theorem~4.13]{admissible}.
\end{proof}

\subsection{Dynamic signatures}
Let us now introduce signatures for the dynamical part of an
operational semantics.  We want a dynamic signature to be something
like an endofunctor on $σ\Trans$, with built-in structuralness.  For
this, we introduce a variant of $ℍ$-transition systems called diplopic
$ℍ$-transition systems, which feature an object of distinguished
vertices, among which all sources of transitions must lie.  This will
enable structuralness, by allowing sources of conclusions of
transition rules to have a distinguished head constructor.  We fix
\anenhanced syntax $σ = (Σ,Γ,{δ∶TS → ST})$ for this subsection.
\begin{definition}
  A \alert{diplopic $ℍ$-transition system} $G$ consists of a
  \alert{vertex} object $V_G ∈ \psh[𝕍𝕋]$, a \alert{distinguished
    vertex} object $D_G ∈ \psh[𝕍𝕋]$, an \alert{edge} object
  $E_G ∈ \psh[𝔼𝕋]$, together with morphisms $γ_G∶ D_G → V_G$ and
  $∂_G∶ E_G → Δ_\source(D_G) × Δ_{\labels,\but}(V_G)$.

  A \alert{diplopic $σ$-transition system} is a diplopic
  $ℍ$-transition system $G$ equipped with $σ$-algebra structure on
  $V_G$.

  As before, we organise both notions into categories
  $ℍ\Trans_𝟚 = \psh[𝔼𝕋]/Δ_𝟚$ and
  $σ\Trans_𝟚 = ℍ\Trans_𝟚 ×_{\psh[𝕍]}\, σ\Alg$, where $Δ_𝟚$
  denotes the composite
  $\psh[𝕍𝕋]^𝟚 \xto{⟨π₁,π₂,π₂⟩} \psh[𝕍𝕋]³ \xto{Δ_\source × Δ_\labels ×
    Δ_\but} \psh[𝔼𝕋]$.
\end{definition}

\begin{definition}\label{def:dynsig}
  A \alert{dynamic signature} over $σ$ is a functor
  $Σ₁∶ σ\Trans → σ\Trans_𝟚$ such that, for all $G ∈ σ\Trans$,
  $V_{Σ₁(G)} = V_G$, $D_{Σ₁(G)} = V_G + Σ(V_G)$, and
  $γ_G∶ V_G + Σ(V_G) → V_G$ is the canonical morphism (and similarly
  on morphisms).
\end{definition}
\begin{example}\label{ex:srciii}
  Letting $𝐝$ denote the syntactic signature of \cref{ex:srcii}. The
  transition rules of \cref{serguei} define a dynamic signature
  $Σ₁∶ σ(𝐝)\Trans → σ(𝐝)\Trans_𝟚$. Its behaviour on the underlying
  $σ(𝐝)$-algebra is fixed, so we merely need to define it on
  transitions.  For any $G = (D,V,E,∂) ∈ σ(𝐝)\Trans$ and $α ∈ 𝔼𝕋$, we
  define $Σ₁(G)(α)$ to be a coproduct over all rules $ρ$ producing a
  transition of type $α$, of a set describing the premises of $ρ$.
  One non-trivial rule is \textsc{(SA)}, whose set of premises is
  $E[𝐜] ×_{V(0_𝐜)} (V(0_𝐜)×V(0_𝐯))$.  Concretely, it is the set of
  tuples $(r,(E,v))$, where $r$ is a transition $e₁ \xto{E'} e₂$, and
  the pullback condition imposes $E' = E[v\ ▫]$.  (We take the
  pullback of
  $E[𝐜] \xto{π₂∂} {V(0_𝐜)} \xot{E[v\ ▫]\ ↤\ E,v} V(0_𝐜)×V(0_𝐯)$.)  We
  define the source of $(r,(E,v))$ to be
  $in₂(in₂ (ι(v),e₁)) ∈ V_𝐩(0)+Σ₀(V)_𝐩(0)$, (i.e., recalling $Σ₀$
  from~\cref{ex:srcalg}, the formal application $ι(v)\ e₁$,) its label
  to be $E$, and its target to be $e₂$.
\end{example}

Returning to the abstract setting, let us now define the category of
models of a dynamic signature $Σ₁$.
For this, we need to build an endofunctor out of $Σ₁$,
hence a link between $σ\Trans$ and $σ\Trans_𝟚$.
\begin{definition}
  Let $ι\Trans∶ σ\Trans → σ\Trans_𝟚$ map any $E → Δ(V)$ to itself
  (with underlying arrow $V → V$).
\end{definition}

\begin{proposition}\label{prop:iota:rho:reflection}
  The functor $ι\Trans∶ σ\Trans → σ\Trans_𝟚$ is a (full) reflective
  embedding.  The left adjoint, say $ρ\Trans∶ σ\Trans_𝟚 → σ\Trans$
  maps any $E → Δ_𝐬(D)×Δ_{𝐥,𝐭}(V)$ to the composite
  $E → Δ_𝐬(D)×Δ_{𝐥,𝐭}(V) → Δ(V)$.
\end{proposition}

\begin{definition}
  For any $Σ₁$, let $\check{Σ}₁$ be the composite
  $σ\Trans \xto{Σ₁} σ\Trans_𝟚 \xto{ρ\Trans} σ\Trans$.
\end{definition}
Models of $Σ₁$ will almost be $\check{Σ}₁$-algebras.  The problem is
that a $\check{Σ}₁$-algebra structure includes in particular algebra
structure for the action of $\check{Σ}₁$ on the underlying
$σ$-algebra, i.e., algebra structure $V → V$ for the identity
endofunctor on $σ$-algebras.  This structure is not relevant for our
purposes, so we require it to be the canonical candidate, i.e., the
identity on $V$.
\begin{definition}
  A $\check{Σ}₁$-algebra structure $\check{Σ}₁(G) → G$ is
  \alert{vertical} when its image under the forgetful functor
  $σ\Trans → σ\Alg$ is the identity.  A $\check{Σ}₁$-algebra is called
  \alert{vertical} accordingly. Let $\check{Σ}₁\alg ᵥ$ denote the full
  subcategory of $\check{Σ}₁\alg$ spanning all vertical algebras.
\end{definition}

\begin{theorem}\label{thm:Z}
  The forgetful functor $\check{Σ}₁\alg ᵥ → σ\Trans$ is monadic, and
  the initial $\check{Σ}₁$-algebra, say $𝐙_{Σ₁}$, or $𝐙$ for short
  when $Σ₁$ is clear from context, may be chosen to be vertical, hence
  in particular to also be initial in $\check{Σ}₁\alg ᵥ$.  (In this
  case, $V_𝐙$ is an initial $σ$-algebra.)
\end{theorem}
\begin{proof}
  Same as \cite[Theorem 5.18 and Proposition 5.19]{HL}.
\end{proof}

\begin{example}\label{ex:srciv}
  For $Σ₁$ as in \cref{ex:srciii}, $𝐙$ is the syntactic transition
  system of~\cref{serguei}.
\end{example}

Let us now collect the static and dynamic part of signatures and their
models.
\begin{definition}
  An \alert{operational signature} consists of a syntactic signature
  $𝐝$, together with a dynamic signature $Σ₁∶ σ(𝐝)\Trans → σ(𝐝)\Trans_𝟚$
  over the generated \enhanced syntax $σ(𝐝)$ (\cref{prop:sigmad}).
  The category of $(𝐝,Σ₁)$-algebras is $\check{Σ}₁\alg ᵥ$.
\end{definition}
By definition, we have:
\begin{proposition}
  The initial vertical $\check{Σ}₁$-algebra is an initial
  $(𝐝,Σ₁)$-algebra.
\end{proposition}

\subsection{Congruence of \enhanced bisimilarity}\label{s:substclosedbisim}
In this subsection, we state our main congruence result.  For this, we
need to make an important hypothesis involving so-called functional
flexible bisimulations. These are like a functional version of
bisimulations, where labels are required to be related instead of
identical, much as in Sangiorgi's
BA-bisimulation~\cite{DBLP:conf/fsen/SangiorgiKS07}, which we need to
define both for algebraic transition systems and their diplopic
variant.  The hypothesis will then require that the considered dynamic
signature $Σ₁$ preserve functional flexible bisimulations.  We again
fix a Howe context $ℍ = (𝕍𝕋,𝔼𝕋,𝐬,𝐭,\labels)$ and \anenhanced syntax
$σ$ over it.
\begin{definition}
  A morphism $f∶ R → X$ in $ℍ\Trans_𝟚$ is a \alert{functional flexible
    bisimulation} iff for any $α ∈ 𝔼𝕋$, $r ∈ D_R(\source(α))$,
  $(r₁,…,r_{n_α}) ∈ Δ_\labels(V_R)(α)$, and transition
  $e'∶ f_D(r) \xto{α(f_V(r₁),…,f_V(r_{n_α}))} x'$ there exists
  $e∶ r \xto{α(r₁,…,r_{n_α})} r'$ such that $f_E(e) = e'$.

  A morphism in $σ\Trans_𝟚$ is a functional flexible bisimulation iff
  the underlying morphism in $ℍ\Trans_𝟚$
  is. 
A morphism in $σ\Trans$ is a functional flexible bisimulation
iff its embedding into $σ\Trans_𝟚$ (by $ι\Trans$) is.  In any of these
categories $𝒞$, let $𝐅𝐅𝐁𝐢𝐬𝐢𝐦(𝒞)$ denote the class of all functional
flexible bisimulations.
\end{definition}

\begin{definition}
  A dynamic signature $Σ₁∶ σ\Trans → σ\Trans_𝟚$ \alert{preserves
    functional flexible bisimulations} iff for all morphisms $f$ in
  $σ\Trans$, if $f$ is a functional flexible bisimulations, then so is
  $Σ₁(f)$.  
\end{definition}

Let us introduce a last hypothesis before stating the main result:
\begin{definition}
  A functor is \alert{algebraic} iff it is finitary and preserves wide
  pullbacks and reflexive coequalisers.  A syntactic signature
  $(Σ,Γ,δ)$ is algebraic if the endofunctor $Σ$ is.
\end{definition}
\begin{remark}
  Algebraicity is straightforward to verify in all our applications.
\end{remark}
\begin{theorem}\label{thm:main}
  For any operational signature $(𝐝,Σ₁)$, if $𝐝$ is algebraic and $Σ₁$
  preserves functional flexible bisimulations, then \enhanced
  bisimilarity on the initial vertical $\check{Σ}₁$-algebra is a
  congruence.
\end{theorem}
\begin{proof}
  See Appendix~\ref{app:proof-main-thm}.
\end{proof}
\subsection{Preservation of functional flexible bisimulations}
In this section, we exhibit a sufficient condition for a dynamic
signature to preserve functional flexible bisimulations, slightly
generalising~\cite[§7]{HL}.  Fixing a Howe context
$ℍ = (𝕍𝕋,𝔼𝕋,\source,\but,\labels)$ and \anenhanced syntax
$σ=(Σ₀,Γ,{δ∶ TS → ST})$ on $\psh[𝕍𝕋]$, we first characterise $ℍ\Trans$
and $ℍ\Trans_𝟚$ as presheaf categories, which allows us to
characterise functional flexible bisimulations as the right class of a
weak factorisation system~\cite{Hovey,riehl} -- we call the left class
\alert{cofibrations}. We then recall familial functors, and define the
notion of rule of a dynamic signature $Σ₁$, and the \alert{border
  arity} of any rule. We finally show that a familial $Σ₁$ preserves
functional flexible bisimulations iff the border arities of all rules
are cofibrations.

Let us characterise transitions systems as presheaves,
recalling~\cref{def:collage}:
\begin{proposition}
  We have $ℍ\Trans ≃ \psh[{𝕍𝕋[𝔼𝕋]}]_{𝐲_𝐬+𝐥+𝐲_𝐭}$, where $𝐲_𝐬+𝐥+𝐲_𝐭∶ 𝔼𝕋 → \psh[𝕍𝕋]$.
\end{proposition}
\begin{proof}
  The functor $Δ_ℍ$ is the nerve of $𝐲_𝐬+𝐥+𝐲_𝐭$, so we conclude
  by~\cite[Lemma~4.9]{DBLP:journals/mscs/CarboniJ95}.
\end{proof}

Doing the same for $ℍ\Trans_𝟚$ leads to considering the functor
$𝔼𝕋 → \psh[𝕍𝕋]^𝟚$ mapping any $α$ to the arrow
$𝐲_{𝐬(α)} → 𝐲_{𝐬(α)} + 𝐥(α) + 𝐲_{𝐭(α)}$. But for
\cite[Lemma~4.9]{DBLP:journals/mscs/CarboniJ95} to apply, we need the
codomain of this functor to be a presheaf category. This is in fact
the case up to equivalence:
\begin{lemma}\label{lem:vtvt}
  We have
  $\psh[𝕍𝕋]^𝟚 ≃ \psh[{𝕍𝕋[𝕍𝕋]}_𝐲]$, where
  $𝐲∶ 𝕍𝕋 → \psh[𝕍𝕋]$.
\end{lemma}
\begin{notation}
  For each state type $b ∈ 𝕍𝕋$, the category $𝕍𝕋[𝕍𝕋]$ has an object
  $b_V$ corresponding to the vertex object, an object $b_D$ for the
  distinguished vertex object, and a morphism $b_V →
  b_D$. 
\end{notation}
Gluing along the obtained functor
$α↦ 𝐲_{𝐬(α)_D} + ∑_{i ∈ n_α} 𝐲_{(𝐥^αᵢ)_V} + 𝐲_{𝐭(α)_V}$, we obtain:
\begin{proposition}
  We have $ℍ\Trans_𝟚 ≃\psh[{𝕍𝕋[𝕍𝕋][𝔼𝕋]}]$.
\end{proposition}
Let us now characterise functional flexible bisimulations by a lifting
property.
\begin{definition}
  In a category $𝒞$, given a class $𝕁$ of morphisms, let
  $\wbotright{𝕁}$ consist of morphisms $f∶ X → Y$ such that
  for any $j∶ A → B$ in $𝕁$, any $(u,v)∶ j → f$ in $𝒞^𝟚$ admits a
  \alert{lifting}, i.e., a morphism $k∶ B → X$ such that $k∘j=u$ and
  $f∘k=v$.  Let $\wbotleft{𝕁}$ consist of all $f$ such that
  any $(u,v)∶ f → j$ in $𝒞^𝟚$ admits a lifting.  A
  \alert{$𝕁$-cofibration} is an element of $\wbotrightleft{𝕁}$.
\end{definition}
\begin{proposition}\label{prop:cx}
    For any $𝕁$,$𝕁$-cofibrations are closed under cobase change and
     composition.
  \end{proposition}

For any $α ∈ 𝔼𝕋$, the element
$(in₁ (\id_{𝐬(α)})) ∈ (𝐲_𝐬+𝐥+𝐲_𝐭)(α)(𝐬(α))$, corresponds to a morphism
$s_α∶ 𝐬(α) → α$ in $𝕍𝕋[𝔼𝕋]$, and similarly we get morphisms
$l^αᵢ∶ 𝐥^αᵢ → α$ for all $i ∈ n_α$.
\begin{definition}
  Let $𝕁_σ$ denote the set of all maps $ℒ'(j_α)$ in $σ\Trans$, where
  $ℒ'∶ \psh[{𝕍𝕋[𝔼𝕋]}] → σ\Trans$ is left adjoint to the forgetful
  functor, and $j_α∶ 𝐲_{𝐬(α)} + ∑_{i ∈ n_α} 𝐲_{𝐥^αᵢ} → 𝐲_α$ denotes
  the cotupling $[𝐲_{s_α}, [𝐲_{l^αᵢ}]_{i ∈ n_α}]$, for all $α$.

  Let $𝕁_{𝟚,σ}$ denote the set of all maps $ℒ'_𝟚(j_{𝟚,α})$ in
  $σ\Trans_𝟚$, where $ℒ'_𝟚∶ \psh[{𝕍𝕋[𝕍𝕋][𝔼𝕋]}] → σ\Trans_𝟚$ is left
  adjoint to the forgetful functor, say $𝒰'_𝟚$, and
  $j_{𝟚,α}∶ 𝐲_{𝐬(α)_D} + ∑_{i ∈ n_α} 𝐲_{(𝐥^αᵢ)_V} → 𝐲_α$ denotes the
  analogous cotupling $[𝐲_{s_{𝟚,α}}, [𝐲_{l^{𝟚,α}ᵢ}]_{i ∈ n_α}]$, for
  all $α$.
\end{definition}
\begin{proposition}
  \label{prop:ffbisim-fib}
  We have $𝐅𝐅𝐁𝐢𝐬𝐢𝐦(σ\Trans) = \wbotright{𝕁_σ}$ and
  $𝐅𝐅𝐁𝐢𝐬𝐢𝐦(σ\Trans_𝟚) = \wbotright{𝕁_{𝟚,σ}}$.
\end{proposition}
Let us now introduce border arities.  A functor $F∶ 𝒞 → \psh[𝔻]$ to
some presheaf category is \alert{familial} iff there exists a functor
$E∶ \el(F(1)) → 𝒞$ from the \alert{category of
  elements}~\cite[§I.5]{MM} of $F(1)$, called the \alert{exponent} of
$F$, such that, we have a natural isomorphism
$$F(C)(d) ≅ ∑_{o ∈ F(1)(d)} 𝒞(E(d,o),C).$$  Intuitively, elements
$o ∈ F(1)(d)$ are operations of output arity $d$, and $E(d,o)$ gives
their input arity.  Morphisms $u∶ d → d'$ of $𝔻$ act on $F(C)$ by
precomposition: for any $o' ∈ F(1)(d')$, we have a morphism
$(d,o) \xto{u↾o'} (d',o')$ in $\el(F (1))$, where $o = F(1)(u)(o')$ --
which we write $o = o'·u$; and the map $F (C) (u)∶ F(C)(d') → F(C)(d)$
sends any $(o',ϕ∶ E(d',o') → C)$ to
$(o, E(d,o) \xto{u↾o'} E(d',o') \xto{ϕ} C)$.  This is the basis for
defining border arities.
\begin{definition}\label{def:border}
  Consider a dynamic signature $Σ₁$ such that the composite
  $σ\Trans \xto{Σ₁} σ\Trans_𝟚 \xto{𝒰'_𝟚} \psh[{𝕍𝕋[𝕍𝕋][𝔼𝕋]}]$ is
  familial with exponent $E$.  Let us fix $α ∈ 𝔼𝕋$ and
  $r ∈ 𝒰'_𝟚Σ₁(1)(α)$.  For any $k∶ Aₖ → α$ among
  $I_α ≔ \{ s_{𝟚,α},l^{𝟚,α}₁,…,l^{𝟚,α}_{n_α}\}$, we have
  $E(k ↾ r)∶ E(Aₖ,r·k) → E(α,r)$.  The \alert{border arity} $𝐛ᵣ$ of
  $r$ is the cotupling
  $[E(k↾r)]_{k ∈ I_α}∶ ∑_{k ∈ I_α} E(A,r·k) → E(α,r)$.
\end{definition}

\begin{theorem}\label{thm:cellular}
  For any dynamic signature $Σ₁∶ σ\Trans → σ\Trans_𝟚$ such that
  $𝒰'_𝟚Σ₁$ is familial, $Σ₁$ preserves functional flexible
  bisimulations iff all border arities are $𝕁_σ$-cofibrations.
\end{theorem}
\begin{proof}[Proof sketch for ``if'', see~\cref{app:proof-cellular}]
  Consider any $(u,v)∶ ℒ'_𝟚(j_{𝟚,α}) → Σ₁(f)$, with $f∶ A → B$ in
  $𝐅𝐅𝐁𝐢𝐬𝐢𝐦(σ\Trans)$. By adjunction, we get
  $(\tilde{u},\tilde{v})∶ j_{𝟚,α} → 𝒰'_𝟚(Σ₁(f))$.  Letting $r$ be the
  composite
  $𝐲_α \xto{\tilde{v}} 𝒰'_𝟚(Σ₁(B)) \xto{𝒰'_𝟚(Σ₁(!))} 𝒰'_𝟚(Σ₁(1))$, we
  use familiality to factor $(\tilde{u},\tilde{v})$ as the solid part
  below.  The result follows from finding $k$ as shown, by
  $𝐛ᵣ ∈ \wbotrightleft{𝕁_σ}$ and $f ∈ \wbotright{𝕁_σ}$.
  \begin{center}
    \hfil
    \diag|baseline=(m-2-1.base)|(.5,1.5){%
      𝐲_{𝐬(α)_D} + ∑_{i ∈ n_α}  𝐲_{(𝐥^αᵢ)_V} \& 𝒰'_𝟚(Σ₁(∑_{k ∈ I_α} E(Aₖ,r·k))) \& 𝒰'_𝟚(Σ₁(A)) \\
      𝐲_α \& 𝒰'_𝟚(Σ₁(E(α,r))) \& 𝒰'_𝟚(Σ₁(B)) %
    }{%
      (m-1-1) edge[loinbat={[𝒰'_𝟚(Σ₁(inₖ)) ∘ (r·k,\id)]_{k ∈ I_α}}{.25}] (m-1-2) %
      edge[labell={j_{𝟚,α}}] (m-2-1) %
      (m-2-1) edge[labelb={(r,\id)}] (m-2-2) %
      (m-2-2) edge[dashed,labelalat={𝒰'_𝟚(Σ₁(k))}{.3}] (m-1-3) %
      (m-1-2) edge[labell={𝒰'_𝟚(Σ₁(𝐛ᵣ))}] (m-2-2) %
      (m-1-2) edge[labela={𝒰'_𝟚(Σ₁(ψ))}] (m-1-3) %
      (m-2-2) edge[labelb={𝒰'_𝟚(Σ₁(ϕ))}] (m-2-3) %
      (m-1-3) edge[labelr={𝒰'_𝟚(Σ₁(f))}] (m-2-3) %
    } \qedhere
  \end{center}
\end{proof}

\begin{example}
  Let us now sketch a proof of \cref{thm:serguei}.  By
  \cref{thm:main,thm:cellular,prop:cx}, it suffices to reconstruct the
  border arity of each rule.  We only treat rule \textsc{(SA)} for
  lack of space: its border arity is the bottom morphism in
\begin{center}
  \Diag{%
    \pbk{m-2-1}{m-2-2}{m-1-2} %
  }{%
    ℒ(0_𝐩 + 0_𝐜) \& ℒ[𝐜] \\
    ℒ(0_𝐯 + 0_𝐩 + 0_𝐜) \& A \rlap{,}
  }{%
    (m-1-1) edge[labela={ℒ[s_{[𝐜]},l_{[𝐜]}]}] (m-1-2) %
    edge[labell={(e₁,E[v\ ▫])}] (m-2-1) %
    (m-2-1) edge[labelb={}] (m-2-2) %
    (m-1-2) edge[labelr={}] (m-2-2) %
  }
\end{center}
with hopefully clear notation.
\end{example}
\begin{example}
  This also works for PCF as in~\cite{DBLP:journals/tcs/Gordon99},
  which we omit for lack of space.
\end{example}

\section{Conclusion and perspectives}\label{sconclu}
We have introduced a categorical framework for applicative
bisimilarity in the presence of operations on terms other than
substitution, and of terms as labels. We have furthermore provided a
notion of signature for generating instances of this framework, and
proved that under suitable hypotheses, notably preservation of
functional flexible bisimulations, applicative bisimilarity in the
generated instance is a congruence. We have finally exhibited a more
concrete sufficient condition in terms of border arities being
cofibrations, which has allowed us to recover congruence of
applicative bisimilarity for $λ$-calculus with delimited control
operators and PCF.

For future work, we would be interested in further generalising the
framework to cover a kind of adaptation of Howe's method that still
eludes our abstraction efforts, namely (early-style) higher-order
process calculi~\cite{DBLP:conf/concur/LengletS15}.

\bibliography{bib}

\newpage

\appendix

\section{Proof of Theorem~\ref{thm:main}}
\label{app:proof-main-thm}

We assume given a Howe context $ℍ = (𝕍𝕋,𝔼𝕋,\source,\but,\labels)$.
To ease readability, we introduce some notations.
\begin{notation}\label{not:diplopic}
    For any $G = (D,V,E,γ,∂) ∈ ℍ\Trans_𝟚$, we let $G_{D,V}$ denote the
    underlying triple $(D,V,{γ∶ D → V}) ∈ \psh[𝕍𝕋]^𝟚$,  $G₀$ denote $V$, $G₁$
    denote $E$, and $Gₛ$ denote $D$.  Furthermore,
    following \cref{not:Delta}, we denote, e.g., by $Δ_{𝟚,𝐬,𝐥}$ the
    functor mapping any $γ∶ D → V$ to $Δ_𝐬(D) × Δ_𝐥(V)$. Finally, we
    sometimes treat the projection $G ↦ G_{D,V}$ as an implicit
    coercion. E.g., we write $Δ_{𝟚,𝐬,𝐥}(G)$ for $Δ_𝐬(D) × Δ_𝐥(V)$.
  \end{notation}

\subsection{Basic properties of flexible bisimulation}
In this section, we establish basic properties of flexible
bisimulations.
  \begin{proposition}\label{prop:Dl:ra}
    The functor $Δ_\labels$ is a right adjoint, hence in particular it
    preserves all limits.
  \end{proposition}
  \begin{proof}
  The functor $Δ_\labels$ is the nerve functor of
  $\labels$.  It is right adjoint to the left Kan extension
  of $\labels$ along the Yoneda embedding, as in the following diagram.
  \begin{center}
    \begin{tikzcd}[ampersand replacement=\&]
      {𝔼𝕋} \&\&\&\& {\hat{𝔼𝕋}} \\
      \\
      \&\& {\hat{𝕍𝕋}}
      \arrow["{𝐥}"', from=1-1, to=3-3]
      \arrow["{𝐲}", from=1-1, to=1-5]
      \arrow[""{name=0, anchor=center, inner sep=0}, "{\bar{𝐥}}"', curve={height=12pt}, from=1-5, to=3-3]
      \arrow[""{name=1, anchor=center, inner sep=0}, "{Δ_𝐥}"', curve={height=12pt}, from=3-3, to=1-5]
      \arrow["\dashv"{anchor=center, rotate=-45}, draw=none, from=0, to=1]
    \end{tikzcd}
  \end{center}
  \end{proof}
  Regarding preservation of colimits, the fact that any $\labels(c)$
  is a finite coproduct of representables entails:
  \begin{proposition}\label{prop:Deltal:finitary}
    The functor $Δ_\labels$ is algebraic, and preserves epimorphisms.
  \end{proposition}
  \begin{proof}
    Just for making the proof slicker, we rely on the well-known
    facts~\cite{algebraictheories} that in presheaf categories
    preserving filtered colimits and reflexive coequalisers is
    equivalent to preserving \emph{sifted} colimits. Furthermore, just
    as the covariant hom of any finitely presentable object preserves
    filtered colimits, in a presheaf category the covariant hom of any
    finite coproduct of representable objects preserves sifted
    colimits, hence epimorphisms. The latter fact deals with the second statement.

    For the first, for any sifted colimit $\colimᵢ Xᵢ$ and $c ∈ \psh$:
    \begin{center}
      \hfill $\begin{array}[b]{rcll}
                Δ_\labels(\colimᵢ Xᵢ)(c) & = & \psh[𝕍𝕋](\labels(c),\colimᵢ Xᵢ) \\
                                         & = & \colimᵢ
                                               \psh[𝕍𝕋](\labels(c),Xᵢ)
                                         & \mbox{($\labels(c)$ a finite coproduct of representables)} \\
                                         & = & \colimᵢ Δ_\labels(Xᵢ)(c).
       \end{array}$ \qedhere
     \end{center}
   \end{proof}

   \begin{proposition}\label{prop:algra}
     All functors $Δ, Δ_𝐥, Δ_𝐬, Δ_𝐭, Δ_{𝐬,𝐥},…$ are algebraic right
     adjoints (and preserve epimorphisms).
   \end{proposition}
   \begin{proof}
     Let us first deal with algebraicity. Because algebraic functors
     are closed under pointwise products, it suffices to deal with
     each of $Δ_𝐥$, $Δ_𝐬$, and $Δ_𝐭$ in isolation: $Δ_𝐬$ and $Δ_𝐭$
     are, as restriction functors; $Δ_𝐥$ is by
     Proposition~\ref{prop:Deltal:finitary}. Finally, in presheaf
     categories, being algebraic entails preservation of epimorphisms.

     For right adjointness, as right adjoints are closed under
     pointwise products (under (co)completeness conditions satisfied
     here), it suffices to show that each of $Δ_𝐥$, $Δ_𝐬$, and $Δ_𝐭$
     is a right adjoint.  Again, $Δ_𝐬$ and $Δ_𝐭$ are, as restriction
     functors; and $Δ_𝐥$ is by Proposition~\ref{prop:Dl:ra}.
   \end{proof}

   \begin{proposition}\label{prop:algra2}
     All functors $Δ_𝟚, Δ_{𝟚,𝐥}, Δ_{𝟚,𝐬}, Δ_{𝟚,𝐭}, Δ_{𝟚,𝐬,𝐥},…$ are
     algebraic right adjoints and preserve epimorphisms.
   \end{proposition}
   \begin{proof}
     Algebraic functors between presheaf categories automatically
     preserve epimorphisms, so it suffices to prove that all these
     functors are algebraic right adjoints.
     
     Algebraic right adjoints being closed under pointwise finite
     products, it further suffices to prove that each of $Δ_{𝟚,𝐥}$,
     $Δ_{𝟚,𝐬}$, and $Δ_{𝟚,𝐭}$ is an algebraic right adjoint.  Now each
     of these functors $Δ_{𝟚,x}$ is the corresponding functor $Δₓ$,
     precomposed with one of the projections $\psh[𝕍𝕋]^𝟚 →
     \psh[𝕍𝕋]$. But each $Δₓ$ is an algebraic right adjoint by
     Proposition~\ref{prop:algra}, and projections, being restriction
     functors, are left and right adjoints, hence algebraic right
     adjoints, hence the result.
   \end{proof}

\begin{lemma}
  \label{lem:wpbk}
  In any presheaf category, for any commuting diagram of the form
  \begin{center}
\begin{tikzcd}[ampersand replacement=\&]
	A \& {A'} \& B \\
	{C'} \& C \& D
	\arrow[from=1-1, to=1-2]
	\arrow[from=1-2, to=1-3]
	\arrow[from=1-2, to=2-2]
	\arrow[from=1-3, to=2-3]
	\arrow[from=2-2, to=2-3]
	\arrow[from=1-1, to=2-1]
	\arrow[two heads, from=2-1, to=2-2]
\end{tikzcd}    
  \end{center}
  if the exterior rectangle is a pointwise weak pullback and the
  marked morphism is epi, then so is the right-hand square.
\end{lemma}
\begin{proof}
  Straightforward, using the fact that any morphism $𝐲_c → C$ from
  some representable presheaf lifts to $C'$ because epis are pointwise
  in presheaf categories.
\end{proof}

\begin{proposition}\label{prop:flexible-pw-pbk}
  A morphism $R→X$ of diplopic $ℍ$-transition systems is a functional flexible bisimulation
  iff the following square is a pointwise weak pullback.
  \begin{center}
    \begin{tikzcd}[ampersand replacement=\&]
      {R₁} \& {X_1} \\
      {Δ_𝐬(Rₛ) × Δ_{𝐥}(R_0)} \& {Δ_𝐬(Xₛ) × Δ_{𝐥}(X_0)}
      \arrow[from=1-1, to=1-2]
      \arrow[from=2-1, to=2-2]
      \arrow[from=1-1, to=2-1]
      \arrow[from=1-2, to=2-2]
    \end{tikzcd}
  \end{center}
  \end{proposition}

\begin{lemma}\label{lem:bisim:epi}
  For any morphisms $R \xto{f} S \xto{g} X$ in $ℍ\Trans_𝟚$
  such that $f_{D,V}∶ R_{D,V} → S_{D,V}$ is an epi,
  if $gf$ is a functional flexible bisimulation, then so is $g$.
\end{lemma}
\begin{proof}
  We have a diagram
    \begin{center}
    \diag{%
      R₁ \& S₁ \& X₁ \\
      Δ_{𝟚,𝐬,𝐥} R_{D,V} \& Δ_{𝟚,𝐬,𝐥} S_{D,V} \& Δ_{𝟚,𝐬,𝐥} X_{D,V}\rlap{,} %
    }{%
      (m-1-1) edge[labela={}] (m-1-2) %
      (m-1-1) edge[labela={}] (m-2-1) %
      (m-2-1) edge[labelb={}] (m-2-2) %
      (m-1-2) edge[labela={}] (m-1-3) %
      edge[labell={}] (m-2-2) %
      (m-2-2) edge[labelb={}] (m-2-3) %
      (m-1-3) edge[labelr={}] (m-2-3) %
    }
  \end{center}
  and want to prove that the right-hand square is a pointwise weak pullback,
  knowing that the outer rectangle is one: this follows readily
  by Lemma~\ref{lem:wpbk} and Proposition~\ref{prop:algra2}.
\end{proof}

\begin{corollary}\label{cor:bisim:epi}
  For any $X ∈ ℍ\Trans_𝟚$ and span morphism $f∶ R → S$ in $ℍ\Trans_𝟚/X²$
  such that $f_{D,V}$ is an epi, if $R$ is a (bi)simulation, then so
  is $S$.
\end{corollary}

\begin{proposition}\label{prop:fib}
  The projection functors $ℍ\Trans → \psh[𝕍𝕋]$ and $ ℍ\Trans_𝟚 → \psh[𝕍𝕋]^𝟚$ are Grothendieck
  fibrations.
\end{proposition}
\begin{proof}
  This follows readily from the next lemma.
\end{proof}
\begin{lemma}\label{lem:fib}
  For any functor $F∶ 𝐀 → 𝐁$ to some category $𝐁$ with pullbacks, the
  projection functor $p∶ 𝐁/F → 𝐀$, mapping any object $b → Fa$ to $a$, is
  a Grothendieck fibration.
\end{lemma}
\begin{proof}
  Given any object $x∶ b → Fa$ and morphism $f∶ a' → a$, a cartesian
  lifting is given by the following pullback,
  \begin{center}
    \Diag{%
      \stdpbk
    }{%
      \restr{b}{a'} \& b \\
      Fa' \& Fa %
    }{%
      (m-1-1) edge[labela={\mrestr{x}{f}}] (m-1-2) %
      edge[labell={\restr{x}{f}}] (m-2-1) %
      (m-2-1) edge[labelb={Ff}] (m-2-2) %
      (m-1-2) edge[labelr={x}] (m-2-2) %
    }%
  \end{center}
  cartesianness being ensured by universal property of pullback.
\end{proof}
\begin{definition}\label{def:flexible}
  A span $R → X×Y$ of diplopic $ℍ$-transition systems (resp. diplopic
  $σ$-transition systems for any enhanced syntax $σ$) is a \alert{flexible simulation} if its left-hand leg
  $R → X$ is a functional flexible bisimulation, and a \alert{flexible
    bisimulation} when both of its legs are.


  By convention, for any $R ∈ \psh[𝕍𝕋]^𝟚$ and $X ∈ ℍ\Trans_𝟚$, a span
  $R → X_{D,V}²$ is a \alert{flexible bisimulation} when the cartesian
  lifting $R^⇑ → X²$ (in the sense of Proposition~\ref{prop:fib}) of
  $X$ along $R → X_{D,V}²$ is.
\end{definition}

\begin{proposition}\label{prop:simliftmax}
  If a span $R → X²$ is a flexible (bi)simulation, then
  so is the cartesian lifting $R_{D,V}^⇑ → X²$.
\end{proposition}
\begin{proof}
  By Corollary~\ref{cor:bisim:epi} applied to the span morphism
  $R → R_{D,V}^⇑$.
\end{proof}

\begin{lemma}\label{lem:laxlim:sepimono}
  Consider any pullback-preserving functor $F∶ 𝐀 → 𝐁$ between
  categories with pullbacks and (strong epi-mono)
  factorisations. Then:
  \begin{romanenumerate}
  \item \label{item:laxlim:monos} A morphism $(f,g)$ in $𝐁/F$ is monic
    iff both $f$ and $g$ are.
  \item \label{item:laxlim:sepi} A morphism $(f,g)$ in $𝐁/F$ is a
    strong epi iff both $f$ and $g$ are.
  \item \label{item:laxlim:sepimono} The forgetful functor
    $𝐁/F → 𝐁×𝐀$ creates, hence preserves, (strong epi-mono)
    factorisations.
\end{romanenumerate}
\end{lemma}
\begin{proof}
  First of all, the forgetful functor creates all colimits, and all
  limits that $F$ preserves, hence in particular pullbacks.
  Furthermore, in any category $𝐂$, a morphism $f∶ X → Y$ is monic iff
  its \alert{self square}
  \begin{center}
    \diag{%
      X \& X \\
      X \& Y %
    }{%
      (m-1-1) edge[identity,labela={}] (m-1-2) %
      edge[identity,labell={}] (m-2-1) %
      (m-2-1) edge[labelb={f}] (m-2-2) %
      (m-1-2) edge[labelr={f}] (m-2-2) %
    }
  \end{center}
  is a pullback.  Thus, a morphism $(f,g)$ in $𝐁/F$ is mono iff its
  self square is a pullback, iff the self squares of $f$ and $g$ are
  both pullbacks, iff $f$ and $g$ are both monic. This
  settles~\cref{item:laxlim:monos}.

  We next deal with the `if' part of~\cref{item:laxlim:sepi}, consider
  any diagram like the solid part in
  \begin{center}
\begin{tikzcd}[ampersand replacement=\&]
	b \&\& d \\
	\& {b'} \&\& {d'} \\
	Fa \&\& Fc \\
	\& {Fa'} \&\& {Fc'\rlap{,}}
	\arrow["x"', from=1-1, to=3-1]
	\arrow["e"', two heads, from=1-1, to=2-2]
	\arrow["Fr"', two heads, from=3-1, to=4-2]
	\arrow["f"{pos=0.2}, from=1-1, to=1-3]
	\arrow["y"'{pos=0.2}, from=1-3, to=3-3]
	\arrow["Fg"{pos=0.2}, from=3-1, to=3-3]
	\arrow["m", hook, from=1-3, to=2-4]
	\arrow["{y'}", from=2-4, to=4-4]
	\arrow["Fj"'{pos=0.2}, from=4-2, to=4-4]
	\arrow["Fs", hook, from=3-3, to=4-4]
	\arrow["k", dashed, from=2-2, to=1-3]
	\arrow["Fl", dashed, from=4-2, to=3-3]
	\arrow["{x'}"'{pos=0.7}, crossing over, from=2-2, to=4-2]
	\arrow["h"'{pos=0.2}, crossing over, from=2-2, to=2-4]
\end{tikzcd}
  \end{center}
  where $e$ and $r$ are strong epis and $m$ and $s$ are monos.  By
  orthogonality, we find unique liftings $k$ and $l$ making all four
  triangles commute (without $F$ on the bottom face). It remains to
  show that the vertical, diagonal square commutes: this follows by
  orthogonality using the fact that $Fs$ is monic (because $F$
  preserves pullbacks, hence monos).
  
  For~\cref{item:laxlim:sepimono}, consider any objects
  $x∶ b → Fa$ and $x'∶ b' → Fa'$, and let $f∶ b → b'$ and $g∶ a → a'$
  make the following diagram commute.
  \begin{center}
    \diag{%
      b \& b' \\
      Fa \& Fa' %
    }{%
      (m-1-1) edge[labela={f}] (m-1-2) %
      edge[labell={x}] (m-2-1) %
      (m-2-1) edge[labelb={Fg}] (m-2-2) %
      (m-1-2) edge[labelr={x'}] (m-2-2) %
    }
  \end{center}
  Let now $b \xonto{e} b'' \xinto{m} b'$ and
  $a \xonto{r} a'' \xinto{s} a'$ be (strong epi-mono) factorisations
  of $f$ and $g$, respectively. Because $F$ preserves monos, $Fs$ is a
  mono, hence by orthogonality we find a unique lifting making both
  squares commute in
  \begin{center}
    \diag{%
      b  \& b''  \& b'  \\
      Fa \& Fa'' \& Fa'\rlap{.} %
    }{%
      (m-1-1) edge[onto,labela={e}] (m-1-2) %
      edge[labell={x}] (m-2-1) %
      (m-2-1) edge[onto,labelb={Fr}] (m-2-2) %
      (m-1-2) edge[dashed,labelr={x''}] (m-2-2) %
      (m-1-2) edge[into,labela={m}] (m-1-3) %
      (m-2-2) edge[into,labelb={Fs}] (m-2-3) %
      (m-1-3) edge[labelr={x'}] (m-2-3) %
    }
  \end{center}
  Furthermore, by~\cref{item:laxlim:monos} and~\cref{item:laxlim:sepi},
  this lifting is in fact a (strong epi-mono) factorisation of
  $(f,g)$, as desired. Preservation follows from (strong epi-mono)
  factorisations being unique up to unique isomorphism and existing in
  $𝐁×𝐀$ by hypothesis.

  Finally, for the `only if' part of~\cref{item:laxlim:sepi}: a
  morphism is a strong epi iff the monic part of its (strong epi-mono)
  factorisation is an isomorphism. So given a strong epi $(e,r)$ in
  $𝐁/F$, we compute its (strong epi-mono) factorisations $m∘e'$ and
  $s∘r'$ of $e$ and $r$, respectively.  By~\cref{item:laxlim:sepimono},
  they lift uniquely to a (strong epi-mono) factorisation
  $(m,s)∘(e',r')$ of $(e,r)$ in $𝐁/F$.  But $(e',r')$ is a strong epi
  by~\cref{item:laxlim:sepi}, and so is $(e,r)$ by hypothesis, and
  $(m,s)$ is a mono between them, hence an isomorphism by
  Lemma~\ref{lem:sepi:mor:iso}. Thus, $m$ and $s$ are both
  isomorphisms, and hence $e$ and $r$ are both strong epis, as
  desired.
\end{proof}

\begin{lemma}\label{lem:diplopic:colims}
The forgetful functor
  $$ℍ\Trans_𝟚 → \psh[𝔼𝕋] × \psh[𝕍𝕋]²$$
  creates all colimits and limits, as well as (strong epi)-mono factorisations.
\end{lemma}
\begin{proof}
  This is clear for colimits. For limits, the projection
  $ℍ\Trans_𝟚 → \psh[𝔼𝕋] × \psh[𝕍𝕋]^𝟚$ creates all limits that the functor
  $Δ_𝟚$ preserves (because $ℍ\Trans_𝟚$ is its lax limit), i.e., all of
  them by Proposition~\ref{prop:algra2}.  For (strong epi)-mono
  factorisations, this follows by Lemma~\ref{lem:laxlim:sepimono}.
\end{proof}

\begin{lemma}\label{lem:diplopic:span:colims}
  For any diplopic $ℍ$-transition system $X$, the forgetful functor
  $$ℍ\Trans_𝟚/X² → ℍ\Trans_𝟚 → \psh[𝔼𝕋] × \psh[𝕍𝕋]²$$
  creates all colimits and connected limits.
\end{lemma}
\begin{proof}
  The projection $(ℍ\Trans_𝟚)/X² → ℍ\Trans_𝟚$ creates colimits and
  connected limits, as any projection from a slice category does.
  The result thus follows by Lemma~\ref{lem:diplopic:colims}.
\end{proof}

\begin{lemma}\label{lem:filtered:flex:filtered}
  Flexible bisimulations are closed under filtered colimits in
  $ℍ\Trans_𝟚^𝟚$, i.e., in the arrow category of $ℍ\Trans_𝟚$.
\end{lemma}
\begin{proof}
  Let $(r_∞∶ R_∞ → X_∞) = \colimⱼ (rⱼ∶ Rⱼ → Xⱼ)$ denote the colimit of
  any filtered digaram of functional flexible bisimulations.  By
  Lemma~\ref{lem:diplopic:span:colims} and the fact that colimits are
  pointwise in the arrow category, we have
  \begin{mathpar}
    (R_∞)₀ ≅ \colimⱼ (R (j) ₀) \and (R_∞)ₛ ≅ \colimⱼ (R (j) ₛ) \and
    (R_∞)₁ ≅ \colimⱼ (R (j) ₁)
  \end{mathpar}
  and
  \begin{mathpar}
    (X_∞)₀ ≅ \colimⱼ (X (j) ₀) \and (X_∞)ₛ ≅ \colimⱼ (X (j) ₛ) \and
    (X_∞)₁ ≅ \colimⱼ (X (j) ₁)\rlap{.}
  \end{mathpar}
  Furthermore,
all morphisms
\begin{mathpar}
γ_{R_∞}∶ (R_∞)ₛ → (R_∞)₀ \and
  ∂_{R_∞}∶ (R_∞)₁ → Δ_𝟚 (R_∞) \\
γ_{X_∞}∶ (X_∞)ₛ → (X_∞)₀ \and
  ∂_{X_∞}∶ (X_∞)₁ → Δ_𝟚 (X_∞) \\
  (r_∞)₁∶ (R_∞)₁ → (X_∞)₁ \and (r_∞)ₛ∶ (R_∞)ₛ → (X_∞)ₛ \and (r_∞)₀∶ (R_∞)₀ → (X_∞)₀
\end{mathpar}
are induced by universal property.

Now consider any $p$ and $q$ making
the following diagram commute.
  \begin{center}
\begin{tikzcd}[ampersand replacement=\&]
	{𝐲_c} \&\& {(X_∞)₁} \\
	\\
	{Δ_{𝟚,𝐬,𝐥}(R_∞)} \&\& {Δ_{𝟚,𝐬,𝐥}(X_∞)}
	\arrow["{∂_{X_∞} }", from=1-3, to=3-3]
	\arrow["{Δ_{𝟚,𝐬,𝐥}r_∞}"', from=3-1, to=3-3]
	\arrow["q", from=1-1, to=1-3]
	\arrow["p"', from=1-1, to=3-1]
\end{tikzcd}    
\end{center}
The functor $Δ_{𝟚,𝐬,𝐥}$ is finitary, and the object $𝐲_c$ finitely
presentable, so $p$ factors through some $Δ_{𝟚,𝐬,𝐥}(Rₖ)$, say as $pₖ$,
and  $q$ factors through some $(Xₗ)₁$, say as $qₗ$.
Furthermore, by filteredness, we find $h$ and morphisms $k \xto{f} h \xot{g} l$, so that we may
define $pₕ$ and $qₕ$ as in the following diagram.
\begin{center}
\begin{tikzcd}[ampersand replacement=\&]
	\&\& {(X_h)₁} \\
	\& {(X_l)_1} \\
	{𝐲_c} \&\&\&\& {Δ_{\mathbb{2},𝐥,𝐭}(X_h)} \\
	\& {Δ_{\mathbb{2},𝐬,𝐥}(Rₖ)} \\
	\&\& {Δ_{\mathbb{2},𝐥,𝐭}(R_h)}
	\arrow["{∂_{X_h}}", from=1-3, to=3-5]
	\arrow["{Δ_{\mathbb{2},𝐬,𝐥}((r_h)_{D,V})}"', from=5-3, to=3-5]
	\arrow["{qₕ}", curve={height=-30pt}, from=3-1, to=1-3]
	\arrow["{Δ_{\mathbb{2},𝐬,𝐥}(R_f)}"{pos=0.8}, from=4-2, to=5-3]
	\arrow["{p_k}", from=3-1, to=4-2]
	\arrow["{pₕ}"', curve={height=30pt}, from=3-1, to=5-3]
	\arrow["{(X_g)_1}"', from=2-2, to=1-3]
	\arrow["{q_l}"', from=3-1, to=2-2]
\end{tikzcd}
\end{center}
Because the following diagram commutes,
\begin{center}
\begin{tikzcd}[ampersand replacement=\&]
	\&\&\& {(X_∞)₁} \\
	\&\& {(X_h)₁} \\
	{𝐲_c} \&\& {?} \& {Δ_{\mathbb{2},𝐥,𝐭}(X_h)} \&\& {Δ_{\mathbb{2},𝐥,𝐭}(X_∞)} \\
	\&\& {Δ_{\mathbb{2},𝐥,𝐭}(R_h)} \\
	\&\&\& {Δ_{\mathbb{2},𝐥,𝐭}(R_∞)}
	\arrow["{∂_{X_h}}"{pos=0.8}, curve={height=-6pt}, from=2-3, to=3-4]
	\arrow["{Δ_{\mathbb{2},𝐬,𝐥}((r_h)_{D,V})}"'{pos=1}, curve={height=6pt}, from=4-3, to=3-4]
	\arrow["{qₕ}"', curve={height=-6pt}, from=3-1, to=2-3]
	\arrow["{pₕ}", curve={height=6pt}, from=3-1, to=4-3]
	\arrow[from=3-4, to=3-6]
	\arrow["q", curve={height=-18pt}, from=3-1, to=1-4]
	\arrow["{∂_{X_∞}}", curve={height=-18pt}, from=1-4, to=3-6]
	\arrow["p"', curve={height=18pt}, from=3-1, to=5-4]
	\arrow["{Δ_{\mathbb{2},𝐬,𝐥}(r_∞)}"', curve={height=18pt}, from=5-4, to=3-6]
	\arrow[from=2-3, to=1-4]
	\arrow[from=4-3, to=5-4]
\end{tikzcd}
\end{center}
by filteredness, we find some $j$ and morphism $u∶ h → j$ such that
$Δ_{𝟚,𝐬,𝐥}Xᵤ$ coequalises the question marked parallel pair above.
We then define $pⱼ$ and $qⱼ$ by composition to obtain a commuting diagram
as the following
\begin{center}
\begin{tikzcd}[ampersand replacement=\&]
	\&\&\& {(X_j)₁} \\
	\&\& {(X_h)₁} \\
	{𝐲_c} \&\& {?} \& {Δ_{\mathbb{2},𝐥,𝐭}(X_h)} \&\& {Δ_{\mathbb{2},𝐥,𝐭}(X_j)} \\
	\&\& {Δ_{\mathbb{2},𝐥,𝐭}(R_h)} \\
	\&\&\& {Δ_{\mathbb{2},𝐥,𝐭}(R_j)}
	\arrow["{∂_{X_h}}"{pos=0.8}, curve={height=-6pt}, from=2-3, to=3-4]
	\arrow["{Δ_{\mathbb{2},𝐬,𝐥}((r_h)_{D,V})}"'{pos=1}, curve={height=6pt}, from=4-3, to=3-4]
	\arrow["{qₕ}"', curve={height=-6pt}, from=3-1, to=2-3]
	\arrow["{pₕ}", curve={height=6pt}, from=3-1, to=4-3]
	\arrow["{Δ_{\mathbb{2},𝐥,𝐭}(X_u)}", from=3-4, to=3-6]
	\arrow["{q_j}", curve={height=-18pt}, from=3-1, to=1-4]
	\arrow["{∂_{X_j}}", curve={height=-18pt}, from=1-4, to=3-6]
	\arrow["{p_j}"', curve={height=18pt}, from=3-1, to=5-4]
	\arrow["{Δ_{\mathbb{2},𝐬,𝐥}(r_j)}"', curve={height=18pt}, from=5-4, to=3-6]
	\arrow["{(X_u)_1}"', from=2-3, to=1-4]
	\arrow["{Δ_{\mathbb{2},𝐥,𝐭}(R_u)}"{pos=0.9}, from=4-3, to=5-4]
\end{tikzcd}
\end{center}
(where again the question marked parallel pair may not commute but the exterior does).

We thus obtain a situation like
  \begin{center}
\begin{tikzcd}[ampersand replacement=\&]
	{𝐲_c} \\
	\\
	{(R_j)₁} \&\&\& {(X_j)_1} \\
	\& {(R_∞)₁} \&\&\& {(X_∞)₁} \\
	{Δ_{\mathbb{2},𝐬,𝐥}(Rⱼ)} \&\&\& {Δ_{\mathbb{2},𝐬,𝐥}(Xⱼ)} \\
	\& {Δ_{\mathbb{2},𝐥,𝐭}(R_∞)} \&\&\& {Δ_{\mathbb{2},𝐥,𝐭}(X_∞).}
	\arrow["{∂_{R_∞}}"{pos=0.2}, from=4-2, to=6-2]
	\arrow["{∂_{X_∞}}", from=4-5, to=6-5]
	\arrow["{Δ_{\mathbb{2},𝐬,𝐥}((r_∞)_{D,V})}"', from=6-2, to=6-5]
	\arrow["{(r_∞)₁}"{pos=0.2}, from=4-2, to=4-5]
	\arrow["q", curve={height=-24pt}, from=1-1, to=4-5]
	\arrow[from=5-1, to=6-2]
	\arrow["{p_j}"', curve={height=24pt}, from=1-1, to=5-1]
	\arrow["{∂_{Rⱼ}}", from=3-1, to=5-1]
	\arrow[from=3-1, to=4-2]
	\arrow["m", dashed, from=1-1, to=3-1]
	\arrow["{Δ_{\mathbb{2},𝐬,𝐥}((rⱼ)_{D,V})}"'{pos=0.8}, from=5-1, to=5-4]
	\arrow[from=5-4, to=6-5]
	\arrow["{∂_{Xⱼ}}"'{pos=0.2}, from=3-4, to=5-4]
	\arrow["{(rⱼ)₁}", from=3-1, to=3-4]
	\arrow[from=3-4, to=4-5]
	\arrow["{q_j}"', from=1-1, to=3-4]
	\arrow["p"{pos=0.2}, fore, curve={height=-6pt}, from=1-1, to=6-2]
      \end{tikzcd}
    \end{center}
  But $rⱼ∶ Rⱼ → Xⱼ$ is a functional flexible bisimulation, so we find a mediating arrow
  $m$ as shown.  The composite
  $$𝐲_c \xto{m} (Rⱼ)₁ → (R_∞)₁$$
  finally provides the desired mediating arrow.
\end{proof}

\begin{corollary}\label{cor:filtered:flex:filtered}
  For any diplopic $ℍ$-transition system $X$, flexible bisimulations $R → X²$
  over $X$ are closed under filtered colimits in $ℍ\Trans_𝟚/X²$.
\end{corollary}

\begin{lemma}\label{lem:filtered:flex:compo}
  For any diplopic $ℍ$-transition system $X$, flexible bisimulations $R → X²$
  over $X$ are closed under span composition.
\end{lemma}
\begin{proof}
  We need to show that the square
  \begin{center}
\begin{tikzcd}[ampersand replacement=\&]
	{(R;S)₁} \&\& {X₁} \\
	{Δ_{𝟚,𝐬,𝐥}(R;S)} \&\& {Δ_{𝟚,𝐬,𝐥}X}
	\arrow[from=1-1, to=2-1]
	\arrow[from=1-3, to=2-3]
	\arrow["{π₁}", from=1-1, to=1-3]
	\arrow["{Δ_{𝟚,𝐬,𝐥}π₁}"', from=2-1, to=2-3]
\end{tikzcd}
  \end{center}
  is a pointwise weak pullback.
  By construction, this square factors as
  \begin{center}
\begin{tikzcd}[ampersand replacement=\&]
	{(R;S)₁} \& {R₁} \& {X₁} \\
	{Δ_{𝟚,𝐬,𝐥}(R;S)} \& {Δ_{𝟚,𝐬,𝐥}R} \& {Δ_{𝟚,𝐬,𝐥}X}
	\arrow[from=1-1, to=2-1]
	\arrow[from=1-3, to=2-3]
	\arrow["{Δ_{𝟚,𝐬,𝐥}π₁}"', from=2-1, to=2-2]
	\arrow["{π₁}", from=1-1, to=1-2]
	\arrow["{π₁}", from=1-2, to=1-3]
	\arrow["{Δ_{𝟚,𝐬,𝐥}π₁}"', from=2-2, to=2-3]
	\arrow[from=1-2, to=2-2]
\end{tikzcd}
  \end{center}
  where the right-hand square is a pointwise weak pullback by
  hypothesis.  Now the left-hand square is the left-hand face in the
  following diagram,
\begin{center}
  \Diag{%
    \pbk[2.5em]{m-2-2}{m-1-1}{m-1-3} %
    \pbk[2.5em]{m-4-2}{m-3-1}{m-3-3} %
    \pullbackk{m-3-3}{m-1-3}{m-2-4}{draw,dashed} %
  }{%
	{(R;S)₁} \&\& {S₁} \\
	\& {R₁} \&\& {X₁} \\
	{Δ_{𝟚,𝐬,𝐥}(R;S)} \&\& {Δ_{𝟚,𝐬,𝐥}S} \\
	\& {Δ_{𝟚,𝐬,𝐥}R} \&\& {Δ_{𝟚,𝐬,𝐥}X}
  }{%
    (m-1-1) edge[labela={}] (m-1-3) %
    edge[labell={}] (m-3-1) %
    (m-3-1) edge[labelb={}] (m-3-3) %
    (m-1-3) edge[labelr={}] (m-3-3) %
    (m-2-2) edge[fore,labelaat={π₂}{.2}] (m-2-4) %
    edge[fore,labell={}] (m-4-2) %
    (m-4-2) edge[labelb={Δ_{𝟚,𝐬,𝐥}π₂}] (m-4-4) %
    (m-2-4) edge[labelr={}] (m-4-4) %
    (m-1-1) edge[labela={}] (m-2-2) %
    (m-3-1) edge[labela={}] (m-4-2) %
    (m-1-3) edge[labelar={π₁}] (m-2-4) %
    (m-3-3) edge[labelar={Δ_{𝟚,𝐬,𝐥}π₁}] (m-4-4) %
  }%
  \end{center}
  whose top and bottom faces are pullbacks by
  Lemma~\ref{lem:diplopic:span:colims} and the fact that
  $Δ_{𝟚,𝐬,𝐥}∶ \psh[𝕍𝕋]^𝟚 → \psh[𝔼𝕋]$, being a right adjoint, is continuous.
  Since the right-hand face is a weak pullback by hypothesis,
  so is the left face by \citet[Lemma~9.26, (i),
  then (ii)]{HL}.  We finally conclude by \citet[Lemma~9.26,
  (i)]{HL}.
\end{proof}

\begin{lemma}\label{lem:spancomp:flex}
  For any diplopic $ℍ$-transition system $X$, flexible bisimulations $R → X_{D,V}²$
  over $X$ are closed under span composition.
\end{lemma}
\begin{proof}
  Given any two flexible bisimulations, say $R$ and $S$, we observe
  that there is a projection $π∶ R^⇑;S^⇑ → (R;S)^⇑$ making the
  following diagram commute.
  \begin{center}
    \begin{tikzcd}[ampersand replacement=\&]
      {R^⇑;S^⇑} \\
      \& {(R;S)^⇑} \& {X₁²} \\
      {Δ_𝟚(R);Δ_𝟚(S)} \& {Δ_𝟚(R;S)} \& {Δ_𝟚X²}
      \arrow["{π}"', from=1-1, to=2-2]
      \arrow[from=2-2, to=3-2]
      \arrow[from=3-2, to=3-3]
      \arrow["{X_∂}", from=2-3, to=3-3]
      \arrow[from=2-2, to=2-3]
      \arrow["\lrcorner"{anchor=center, pos=0.125}, draw=none, from=2-2, to=3-3]
      \arrow[from=1-1, to=3-1]
      \arrow["{≅ }"', from=3-1, to=3-2]
      \arrow[curve={height=-12pt}, from=1-1, to=2-3]
    \end{tikzcd}    
  \end{center}
  To see this, we observe that by interchange of limits
  $R^⇑;S^⇑$ is the limit of
  \begin{center}
\begin{tikzcd}[ampersand replacement=\&]
	{X₁} \&\&\&\& {X₁} \&\&\&\& {X₁} \\
	\&\& {Δ_𝟚R} \&\&\&\& {Δ_𝟚S} \\
	{Δ_𝟚X} \&\&\&\& {Δ_𝟚X} \&\&\&\& {Δ_𝟚X}
	\arrow["{X_∂}"', from=1-1, to=3-1]
	\arrow["{Δ_𝟚π₁}", from=2-3, to=3-1]
	\arrow["{Δ_𝟚π₂}"', from=2-3, to=3-5]
	\arrow["{Δ_𝟚π₁}", from=2-7, to=3-5]
	\arrow["{Δ_𝟚π₂}"', from=2-7, to=3-9]
	\arrow["{X_∂}"', from=1-5, to=3-5]
	\arrow["{X_∂}", from=1-9, to=3-9]
\end{tikzcd}    
  \end{center}
  while $(R;S)^⇑$ is the limit of
  the following subdiagram.
  \begin{center}
\begin{tikzcd}[ampersand replacement=\&]
	{X₁} \&\&\&\&\&\&\&\& {X₁} \\
	\&\& {Δ_𝟚R} \&\&\&\& {Δ_𝟚S} \\
	{Δ_𝟚X} \&\&\&\& {Δ_𝟚X} \&\&\&\& {Δ_𝟚X}
	\arrow["{X_∂}"', from=1-1, to=3-1]
	\arrow["{Δ_𝟚π₁}", from=2-3, to=3-1]
	\arrow["{Δ_𝟚π₂}"', from=2-3, to=3-5]
	\arrow["{Δ_𝟚π₁}", from=2-7, to=3-5]
	\arrow["{Δ_𝟚π₂}"', from=2-7, to=3-9]
	\arrow["{X_∂}", from=1-9, to=3-9]
\end{tikzcd}    
  \end{center}
  Finally, by Lemma~\ref{lem:filtered:flex:compo}, we know that $R^⇑;S^⇑$ is
  a flexible bisimulation, hence so is $R;S$ by Lemma~\ref{lem:wpbk}.
\end{proof}

\begin{lemma}\label{lem:sepi:mor:iso}
  For any strong epis $e∶ X → Y$ and $e'∶ X → Z$, any
  mono $f∶ Y → Z$ such that $f∘e = e'$ is an isomorphism.
\end{lemma}
\begin{proof}
  We find a section of $f$ by lifting as in
  \begin{center}
    \diag{%
      X \& Y \\
      Z \& Z\rlap{.} %
    }{%
      (m-1-1) edge[labela={e}] (m-1-2) %
      edge[labell={e'}] (m-2-1) %
      (m-2-1) edge[identity,labelb={}] (m-2-2) %
      (m-1-2) edge[labelr={f}] (m-2-2) %
      (m-2-1) edge[dashed,labelal={l}] (m-1-2) %
    }
  \end{center}
  But $l$ is in fact an inverse by uniqueness of lifting in

  \hfill
\begin{tikzcd}[ampersand replacement=\&,baseline=(\tikzcdmatrixname-3-1.base)]
	X \&\& Y \\
	\& Z \\
	Y \&\& Z\rlap{.}
	\arrow["e"', from=1-1, to=3-1]
	\arrow["e", from=1-1, to=1-3]
	\arrow["{e'}"', from=1-1, to=2-2]
	\arrow["f"', from=3-1, to=3-3]
	\arrow["f", from=1-3, to=3-3]
	\arrow["f", from=3-1, to=2-2]
	\arrow["l", from=2-2, to=1-3]
	\arrow[Rightarrow, no head, from=2-2, to=3-3]
\end{tikzcd}
\end{proof}

\begin{lemma}\label{lem:images:flex}
  For any diplopic $ℍ$-transition system $X$, flexible bisimulations $R → X²$
  over $X$ are closed under images.
\end{lemma}

\begin{proof}
  Both $\psh[𝔼𝕋]$ and $\psh[𝕍𝕋]^𝟚$ are (isomorphic to) presheaf categories,
  hence images are computed as (strong epi-mono) factorisations.
  Furthermore, $Δ_𝟚$ preserves pullbacks by
  Proposition~\ref{prop:algra2}, hence by
  Lemma~\ref{lem:laxlim:sepimono} the forgetful functor
  $ℍ\Trans_𝟚 → \psh[𝔼𝕋] × \psh[𝕍𝕋]^𝟚$ creates (strong epi-mono)
  factorisations, hence images.

  Now, consider any flexible bisimulation $p∶ R → X²$.  As we just
  saw, we obtain a (strong epi-mono) factorisation of $p$ by factoring
  $p₁$ and $p_{D,V}$.  We then need to show that the square
  \begin{center}
\begin{tikzcd}[ampersand replacement=\&]
	{\im(R_1)} \& {X₁} \\
	{Δ_{𝟚,𝐬,𝐥}(\im(R_{D,V}))} \& {Δ_{𝟚,𝐬,𝐥}X}
	\arrow["{Δ_{𝟚,𝐬,𝐥}π₁}"', from=2-1, to=2-2]
	\arrow["{X_∂}", from=1-2, to=2-2]
	\arrow[from=1-1, to=2-1]
	\arrow[from=1-1, to=1-2]
      \end{tikzcd}
    \end{center}
    is a pointwise weak pullback.
    But by Proposition~\ref{prop:algra2}, $Δ_{𝟚,𝐥,𝐬}$ preserves epimorphisms. Thus, since the exterior of
    \begin{center}
\begin{tikzcd}[ampersand replacement=\&]
	{R₁} \& {\im(R_1)} \& {X₁} \\
	{Δ_{𝟚,𝐬,𝐥}(R_{D,V})} \& {Δ_{𝟚,𝐬,𝐥}(\im(R_{D,V}))} \& {Δ_{𝟚,𝐬,𝐥}X}
	\arrow["{Δ_{𝟚,𝐬,𝐥}π₁}"', from=2-2, to=2-3]
	\arrow["{X_∂}", from=1-3, to=2-3]
	\arrow[from=1-2, to=2-2]
	\arrow[from=1-2, to=1-3]
	\arrow[two heads, from=2-1, to=2-2]
	\arrow[from=1-1, to=2-1]
	\arrow[from=1-1, to=1-2]
\end{tikzcd}
\end{center}
is a pointwise weak pullback by hypothesis, we conclude by Lemma~\ref{lem:wpbk}.
\end{proof}

\subsection{Composition of flexible and rigid simulations}

Our goal in this subsection is to prove the following.
\begin{lemma}\label{lem:flex:rigid:compo}
  For any $ℍ$-transition system $X$, diplopic flexible simulation
  $R → X²$, and simulation $S₀ → X₀²$, equipped with a span
  morphism $ρ∶ R₀;S₀ → R₀$, the relation $\im(R; θS₀)$ is a
  flexible simulation, hence so is $\im(R_{D,V};S₀)^⇑$.
\end{lemma}

In order to prove this smoothly, we introduce the following notion of
triplopic transition system.
\begin{definition}
  Let $ℍ\Trans_3$ denote the lax limit of
  $\psh[𝕍𝕋]³ \xto{Δ_𝐬×Δ_𝐥×Δ_𝐭} \psh[𝔼𝕋]$.
  Objects of $ℍ\Trans_3$ are called \alert{triplopic transition systems}.
\end{definition}
\begin{notation}
  We denote by $Δ₃$, $Δ_{3,𝐬}$, $Δ_{3,𝐬,𝐥}$,… the functors analogous
  to $Δ_𝟚$, $Δ_{𝟚,𝐬}$, $Δ_{𝟚,𝐬,𝐥}$,…,
  and often treat the projection $ℍ\Trans₃ → \psh[𝕍𝕋]³$ as an implicit coercion,
  thus writing, e.g., $Δ_{3,𝐬,𝐥}X$ for any $X ∈ ℍ\Trans₃$, meaning
  $Δ_𝐬(Xₛ) × Δ_𝐥(Xₗ)$.
\end{notation}
A triplopic transition system $X$ thus consists of presheaves
$Xₛ,Xₗ,Xₜ ∈ \psh[𝕍𝕋]$ and $X₁ ∈ \psh[𝔼𝕋]$, together with a morphism
$X₁ → Δ_𝐬(Xₛ)×Δ_𝐥(Xₗ)×Δ_𝐭(Xₜ)$.
\begin{remark}
  We use a boldface $𝟚$ in $ℍ\Trans_𝟚$ and a normal $3$ in $ℍ\Trans₃$,
  to reflect the fact that any diplopic transition system
  $X ∈ ℍ\Trans_𝟚$ comes with a morphism $Xₛ → X₀$, while there is no
  such requirement for triplopic transition systems.
\end{remark}

Let us readily notice the following useful facts.
\begin{proposition}\label{prop:algra3}
     All functors $Δ_3, Δ_{3,𝐥}, Δ_{3,𝐬}, Δ_{3,𝐭}, Δ_{3,𝐬,𝐥},…$ are
     algebraic right adjoints and preserve epimorphisms.
   \end{proposition}
   \begin{proof}
     Algebraic functors between presheaf categories automatically
     preserve epimorphisms, so it suffices to prove that all these
     functors are algebraic right adjoints.
     
     Algebraic right adjoints being closed under pointwise finite
     products, it further suffices to prove that each of $Δ_{3,𝐥}$,
     $Δ_{3,𝐬}$, and $Δ_{3,𝐭}$ is an algebraic right adjoint.  Now each
     of these functors $Δ_{3,x}$ is the corresponding functor $Δₓ$,
     precomposed with one of the projections $\psh[𝕍𝕋]^3 →
     \psh[𝕍𝕋]$. But each $Δₓ$ is an algebraic right adjoint by
     Proposition~\ref{prop:algra}, and projections, being restriction
     functors, are left and right adjoints, hence algebraic right
     adjoints, hence the result.
   \end{proof}

\begin{lemma}\label{lem:triplopic:colims}
The forgetful functor
  $$ℍ\Trans_3 → \psh[𝔼𝕋] × \psh[𝕍𝕋]²$$
  creates all colimits and limits, as well as (strong epi)-mono
  factorisations.
\end{lemma}
\begin{proof}
  Just as Lemma~\ref{lem:diplopic:colims}.
\end{proof}

The idea of triplopic transition systems is to unify flexible and
rigid bisimulation into a single framework, while allowing maximal
flexibility in the choice of input and output states, and labels.  Let
us now define (bi)simulation in triplopic transition systems.  We will
then describe embeddings of transition systems and diplopic transition
systems into triplopic transition systems, proving in each case that
the embedding preserves and reflects bisimulation.

\begin{definition}
  A morphism $f∶ R → X$ of triplopic transition systems is a \alert{functional bisimulation}
  iff the square
  \begin{center}
    \diag{%
      R₁ \& X₁ \\
      Δ_{3,𝐬,𝐥}R \& Δ_{3,𝐬,𝐥}X
    }{%
      (m-1-1) edge[labela={}] (m-1-2) %
      edge[labell={}] (m-2-1) %
      (m-2-1) edge[labelb={}] (m-2-2) %
      (m-1-2) edge[labelr={}] (m-2-2) %
    }
  \end{center}
  is a pointwise weak pullback.  Spans and relations in $ℍ\Trans_3$
  are called simulations and bisimulations analogously to the case of
  $ℍ\Trans_𝟚$.
\end{definition}

\begin{proposition}
  Mapping any diplopic transition system
  $$(X₁,γ∶Xₛ→X₀,∂∶X₁ → Δ_𝐬(Xₛ) × Δ_{𝐥,𝐭}(X₀))$$ to
  $$(X₁,Xₛ,X₀,X₀,∂∶X₁ → Δ_𝐬(Xₛ) × Δ_{𝐥,𝐭}(X₀))$$
  yields an embedding $ι∶ ℍ\Trans_𝟚 → ℍ\Trans₃$.
\end{proposition}
\begin{proof}
  Straightforward.
\end{proof}
\begin{notation}
  By composition with $ℍ\Trans ↪ ℍ\Trans_𝟚$, we obtain a further
  embedding $ℍ\Trans ↪ ℍ\Trans₃$. Treating the former as an implicit
  coercion, we thus often also merely denote the composite by $ι$.
\end{notation}

\begin{proposition}\label{prop:triplopic:sim:reflection}
  A morphism (resp.\ a span) of diplopic transition systems is a
  functional bisimulation (resp.\ a simulation or bisimulation) iff
  its embedding into triplopic transition systems is.
\end{proposition}
\begin{proof}
  Straightforward.
\end{proof}

Beyond the embedding $ℍ\Trans ↪ ℍ\Trans₃$ that we saw above,
there is the following embedding of spans:
\begin{proposition}
  For any $X ∈ ℍ\Trans$, mapping any span $R₀ → X₀²$ in $\psh[𝕍𝕋]$ to
  the triplopic transition system $θ(R₀)$ given by $(R₀,X₀,R₀)$ and
  $θ(R₀)₁ = R₀^↑$, i.e., given 
  by the pullback
  \begin{center}
    \Diag{%
      \stdpbk %
    }{%
      R₀^↑ \& X₁² \\
      Δ_𝐬(R₀)×Δ_𝐥(X₀) × Δ_𝐭(R₀) \& ΔX₀²\rlap{,}
    }{%
      (m-1-1) edge[labela={}] (m-1-2) %
      edge[labell={}] (m-2-1) %
      (m-2-1) edge[labelb={}] (m-2-2) %
      (m-1-2) edge[labelr={}] (m-2-2) %
    }
  \end{center}
  extends to an embedding $θ∶ \psh[𝕍𝕋]/X₀² → ℍ\Trans_3/X²$, which we
  call the \alert{thin} embedding.
\end{proposition}
\begin{remark}
  Thinness here refers to labels, which are forced to agree on both
  sides of any transition in $θ(R₀)$.
\end{remark}

The thin embedding enables the following characterisation of
bisimulation in $ℍ$-transition systems in terms of bisimulation in
triplopic $ℍ$-transition systems:
\begin{proposition}
  For any $X ∈ ℍ\Trans$, a span $R₀ → X₀²$ is a simulation (resp.\
  bisimulation) iff $θ(R₀) → X²$ is one.
\end{proposition}
\begin{proof}
  Both statements mean that the square
  \begin{center}
    \diag{%
      R₀^↑ \& X₁ \\
      Δ_𝐬(R₀)×Δ_𝐥(X₀) \& Δ_{𝐬,𝐥}X₀
    }{%
      (m-1-1) edge[labela={π₁}] (m-1-2) %
      edge[labell={}] (m-2-1) %
      (m-2-1) edge[labelb={π₁}] (m-2-2) %
      (m-1-2) edge[labelr={}] (m-2-2) %
    }
  \end{center}
  is a pointwise weak pullback.
\end{proof}

Finally, we have the easy
\begin{proposition}
  (Bi)simulations are closed under span composition in $ℍ\Trans₃$.
\end{proposition}
\begin{proof}
  By symmetry it suffices to show that simulations are closed under
  span composition.  Let us thus consider any simulations $R$ and $S$
  over some $X ∈ ℍ\Trans₃$.
  We must show that the square
  \begin{center}
    \diag{%
      (R;S)₁ \& X₁ \\
      Δ_{3,𝐬,𝐥} (R;S) \& Δ_{3,𝐬,𝐥} X
    }{%
      (m-1-1) edge[labela={π₁}] (m-1-2) %
      edge[labell={}] (m-2-1) %
      (m-2-1) edge[labelb={Δ_{3,𝐬,𝐥}π₁}] (m-2-2) %
      (m-1-2) edge[labelr={}] (m-2-2) %
    }
  \end{center}
  is a pointwise weak pullback.
  This square factors as
  \begin{center}
    \diag{%
      (R;S)₁ \& R₁ \& X₁ \\
      Δ_{3,𝐬,𝐥} (R;S) \& Δ_{3,𝐬,𝐥} (R) \& Δ_{3,𝐬,𝐥} X\rlap{,}
    }{%
      (m-1-1) edge[labela={π₁}] (m-1-2) %
      edge[labell={}] (m-2-1) %
      (m-2-1) edge[labelb={Δ_{3,𝐬,𝐥}π₁}] (m-2-2) %
      (m-1-2) edge[labelr={}] (m-2-2) %
      (m-1-2) edge[labela={π₁}] (m-1-3) %
      (m-2-2) edge[labelb={Δ_{3,𝐬,𝐥}π₁}] (m-2-3) %
      (m-1-3) edge[labelr={}] (m-2-3) %
    }
  \end{center}
  where the right-hand square is a pointwise weak pullback by hypothesis, and
  the left-hand square is the left-hand face in
  \begin{center}
  \Diag{%
    \pullbackk[2.5em][2.5em](6pt){m-2-2}{m-1-1}{m-1-3}{draw,-} %
    \pullbackk[3.5em][3.5em](12pt){m-4-2}{m-3-1}{m-3-3}{draw,-} %
    \pullbackk{m-3-3}{m-1-3}{m-2-4}{draw,dashed} %
  }{%
	{(R;S)₁} \&\& {S₁} \\
	\& {R₁} \&\& {X₁} \\
	{Δ_{𝟛,𝐬,𝐥}(R;S)} \&\& {Δ_{𝟛,𝐬,𝐥}S} \\
	\& {Δ_{𝟛,𝐬,𝐥}R} \&\& {Δ_{𝟛,𝐬,𝐥}X}
  }{%
    (m-1-1) edge[labela={}] (m-1-3) %
    edge[labell={}] (m-3-1) %
    (m-3-1) edge[labelb={}] (m-3-3) %
    (m-1-3) edge[labelr={}] (m-3-3) %
    (m-2-2) edge[fore,labelaat={π₂}{.2}] (m-2-4) %
    edge[fore,labell={}] (m-4-2) %
    (m-4-2) edge[labelb={Δ_{𝟛,𝐬,𝐥}π₂}] (m-4-4) %
    (m-2-4) edge[labelr={}] (m-4-4) %
    (m-1-1) edge[labela={}] (m-2-2) %
    (m-3-1) edge[labela={}] (m-4-2) %
    (m-1-3) edge[labelar={π₁}] (m-2-4) %
    (m-3-3) edge[labelar={Δ_{𝟛,𝐬,𝐥}π₁}] (m-4-4) %
  }%
  \end{center}
  whose top and bottom faces are pullbacks by
  Lemma~\ref{lem:diplopic:span:colims} and the fact that $Δ_{3,𝐬,𝐥}$,
  being a right adjoint, is continuous.  Since the right-hand face is
  a pointwise weak pullback by hypothesis, so is the left-hand face by
  \citet[Lemma~9.26, (i), then (ii)]{HL}.  The whole rectangle thus is
  a pointwise weak pullback by \citet[Lemma~9.26, (i)]{HL}, as desired.
\end{proof}

\begin{proposition}\label{prop:triplopic:images}
  Triplopic (bi)simulations are closed under images.
\end{proposition}
\begin{proof}
  By symmetry it suffices to treat the case of simulations.
  Let $R → X²$ be any triplopic simulation.
  Then by Proposition~\ref{lem:triplopic:colims} we need to prove that
  the right-hand square below is a pointwise weak pullback,
  \begin{center}
    \diag{%
      R₁ \& \im(R₁) \& X₁² \\
      Δ_{3,𝐬,𝐥}R \& Δ_{3,𝐬,𝐥}\im(R) \& Δ_{3,𝐬,𝐥}X² %
    }{%
      (m-1-1) edge[onto,labela={}] (m-1-2) %
      edge[labell={}] (m-2-1) %
      (m-2-1) edge[onto,labelb={}] (m-2-2) %
      (m-1-2) edge[labelr={}] (m-2-2) %
      (m-1-2) edge[into,labela={}] (m-1-3) %
      (m-2-2) edge[into,labelb={}] (m-2-3) %
      (m-1-3) edge[labelr={}] (m-2-3) %
    }
  \end{center}
  which is the case by Lemma~\ref{lem:wpbk} and the fact that
  $Δ_{3,𝐬,𝐥}$ preserves epis by algebraicity
  (Lemma~\ref{prop:algra3}).
\end{proof}

\begin{lemma}\label{lem:triplopic:retraction}
  Given a retraction $R ↠ S$ over any $X²$ in $ℍ\Trans₃$,
  if $S$ is a simulation, then so is $\im(R)$.
\end{lemma}
\begin{proof}
  The given retraction and its section yield morphisms
  \begin{center}
  $\im(R) → \im(S)$ \qquad \qquad and \qquad \qquad $\im(S) → \im(R)$,
\end{center}
hence $\im(R) ≅ \im(S)$, so
  we conclude by Lemma~\ref{prop:triplopic:images}.
\end{proof}

\begin{proof}[Proof of Lemma~\ref{lem:flex:rigid:compo}]
  The morphism
   $\tilde{ρ}∶ R;ιS → R;θS$ defined by the triple
  \begin{mathpar}
    \id∶ R₀;S₀ → R₀;S₀ \and
    ρ∶ R₀;S₀ → R₀ \and
    \id∶ R₀;S₀ → R₀;S₀
  \end{mathpar}
  admits a section, namely the morphism $R;θS → R;ιS$ defined by
  \begin{mathpar}
    \id∶ R₀;S₀ → R₀;S₀ \and
    R₀ ≅ R₀;X₀ → R₀;S₀ \and
    \id∶ R₀;S₀ → R₀;S₀
  \end{mathpar}
  (induced by reflexivity of $S₀$). Thus, $\im(R;ιS)$ is
  a triplopic simulation by Lemma~\ref{lem:triplopic:retraction},
  hence a diplopic one by
  Proposition~\ref{prop:triplopic:sim:reflection}.  Finally,
  $\im(R_{D,V};S₀)^⇑$ is a flexible simulation by
  Proposition~\ref{prop:simliftmax}.
\end{proof}

\subsection{Fundamental property of flexible bisimulation}
In this section, we reduce the theorem to a certain result involving
flexible bisimulations, using the following fundamental property of
flexible bisimulation:
\begin{proposition}\label{prop:flexbisim:bisim}
  For any $X ∈ ℍ\Trans$ and reflexive, flexible bisimulation $R → X²$, $R₀ → X₀²$ is a
  bisimulation.
\end{proposition}

We need the following lemma.

\begin{lemma}\label{lem:damier:wpbk}
  Consider any commuting diagram of the following form
  \begin{center}
\begin{tikzcd}[ampersand replacement=\&]
	A \& B \& C \& D \\
	X \& Y \& Z \\
	T \& U \& V \& W
	\arrow["f", from=1-1, to=1-3, curve={height=-18pt}]
	\arrow["g", from=1-2, to=1-3]
	\arrow["j", from=2-1, to=2-2]
	\arrow["y", from=1-2, to=2-2]
	\arrow["x"', from=1-1, to=2-1]
	\arrow["t"', from=2-1, to=3-1]
	\arrow["l"', from=3-1, to=3-2]
	\arrow["u", from=2-2, to=3-2]
	\arrow["m"', from=3-2, to=3-3]
        \arrow["n"', from=3-3, to=3-4]
	\arrow["h", from=1-3, to=1-4]
	\arrow["w", from=1-4, to=3-4]
	\arrow["z", from=1-3, to=2-3]
	\arrow["k", from=2-2, to=2-3]
        \arrow["v", from=2-3, to=3-3]
      \end{tikzcd}
    \end{center}
    (i.e., all three squares and the rectangle commute, plus $zf=kjx$),
    such that
    all three squares below are weak pullbacks.
    \begin{center}
\begin{tikzcd}[ampersand replacement=\&]
	X \& Y \\
	T \& U
	\arrow["j", from=1-1, to=1-2]
	\arrow["t"', from=1-1, to=2-1]
	\arrow["l"', from=2-1, to=2-2]
	\arrow["u", from=1-2, to=2-2]
\end{tikzcd} \hfil      
\begin{tikzcd}[ampersand replacement=\&]
	A \&  \& C \\
	X \& Y \& Z
	\arrow["f", from=1-1, to=1-3]
	\arrow["j"', from=2-1, to=2-2]
	\arrow["x"', from=1-1, to=2-1]
	\arrow["z", from=1-3, to=2-3]
	\arrow["k"', from=2-2, to=2-3]
      \end{tikzcd}
      \hfil
\begin{tikzcd}[ampersand replacement=\&]
	B \& C \& D \\
	Y \\
	U \& V \& W
	\arrow["g", from=1-1, to=1-2]
	\arrow["y"', from=1-1, to=2-1]
	\arrow["u"', from=2-1, to=3-1]
	\arrow["m"', from=3-1, to=3-2]
        \arrow["n"', from=3-2, to=3-3]
	\arrow["h", from=1-2, to=1-3]
	\arrow["w", from=1-3, to=3-3]
\end{tikzcd}      
    \end{center}
    Then, the exterior is again a weak pullback.
\end{lemma}
\begin{proof}
  First, we find $i∶ A → B$ such that $gi = f$ and $uyi = ltx$, by
  weak universal property of $B$.

  Now, consider any cone $(p,q)$ as shown below.
  \begin{center}
\begin{tikzcd}[ampersand replacement=\&]
	E \\
	\&\& A \& B \& C \& D \\
	\&\& X \& Y \& Z \\
	\&\& T \& U \& V \& W
	\arrow["i"', from=2-3, to=2-4]
	\arrow["g"', from=2-4, to=2-5]
	\arrow["h", from=3-3, to=3-4]
	\arrow["y", from=2-4, to=3-4]
	\arrow["x"', from=2-3, to=3-3]
	\arrow["t"', from=3-3, to=4-3]
	\arrow["l"', from=4-3, to=4-4]
	\arrow["u", from=3-4, to=4-4]
        \arrow["v", from=3-5, to=4-5]
	\arrow["h"', from=2-5, to=2-6]
	\arrow["w", from=2-6, to=4-6]
	\arrow["z", from=2-5, to=3-5]
	\arrow["k", from=3-4, to=3-5]
	\arrow["p", curve={height=-18pt}, from=1-1, to=2-6]
	\arrow["q"', curve={height=18pt}, from=1-1, to=4-3]
	\arrow["r", dashed, from=1-1, to=2-4]
	\arrow["s"', dashed, from=1-1, to=3-3]
	\arrow["gr"{description, pos=0.6}, curve={height=-12pt}, dashed, from=1-1, to=2-5]
	\arrow["m"', from=4-4, to=4-5]
        \arrow["n"', from=4-5, to=4-6]
	\arrow["d"{description, pos=0.7}, dashed, from=1-1, to=2-3]
\end{tikzcd}
  \end{center}
  By weak universal property of $B$, we find $r∶ E → B$ such that
  $hgr = p$ and $vyr = lq$.  By weak universal property of $X$, we
  then find a morphism $s∶ E → X$ such that $us = q$ and $hs = yr$.
  Finally, by weak universal property of $A$, we find the desired
  morphism $d∶ E → A$ such that $xd = s$ and $gid = gr$.  Please note
  that nothing here guarantees that $id = r$, nor that $yi=hx$, but
  this does invalidate the result.
\end{proof}

\begin{proof}[Proof of Proposition~\ref{prop:flexbisim:bisim}]
  By symmetry, it suffices to check that the first projection
  $π₁∶ R₀ → X₀$ is a simulation.  For any $c ∈ 𝔼𝕋$, we form the following diagram,
  \begin{center}
\ajustedroit{
\begin{tikzcd}[ampersand replacement=\&]
	{R_0^\uparrow(c)} \& {R_1(c)} \& {(X_1 \times X_1)(c)} \& {X_1(c)} \\
	{(Δ_𝐬(R_0) × Δ_𝐥(X_0) × Δ_𝐭(R_0))(c)} \& {(Δ_𝐬(R_0) × Δ_𝐥(R_0) × Δ_𝐭(R_0))(c)} \& {(Δ_𝐬(X_0)² × Δ_𝐥(X_0)² × Δ_𝐭(X_0)²)(c)} \\
	{(Δ_𝐬(R_0) × Δ_𝐥(X_0))(c)} \& {(Δ_𝐬(R_0) × Δ_𝐥(R_0))(c)} \& {(Δ_𝐬(X_0)² × Δ_𝐥(X_0)²)(c)} \& {(Δ_𝐬(X_0) × Δ_𝐥(X_0))(c)}
	\arrow[from=1-3, to=2-3]
	\arrow[from=1-1, to=2-1]
	\arrow[from=2-1, to=3-1]
	\arrow[from=2-3, to=3-3]
	\arrow["{\pi_1}", from=1-3, to=1-4]
	\arrow[from=1-4, to=3-4]
	\arrow["{ \pi_1 \times \pi_1}"', from=3-3, to=3-4]
	\arrow[from=2-1, to=2-2]
	\arrow[from=2-2, to=2-3]
	\arrow[from=1-2, to=2-2]
	\arrow[from=1-2, to=1-3]
	\arrow[from=2-2, to=3-2]
	\arrow[from=3-1, to=3-2]
	\arrow[from=3-2, to=3-3]
	\arrow[curve={height=-18pt}, from=1-1, to=1-3]
      \end{tikzcd}
    }
  \end{center}
and conclude by
  Lemma~\ref{lem:damier:wpbk}. To check that it applies, we observe
  that
  \begin{itemize}
  \item the first requirement holds easily (the bottom left square is
    easily seen to be a pullback);
  \item the second requirement holds by construction of $R₀^⇑$; and
  \item the last requirement holds by hypothesis that $R$ is a
    flexible bisimulation. \qedhere
  \end{itemize}
\end{proof}

Let us now use the fundamental property
(Proposition~\ref{prop:flexbisim:bisim}) of flexible bisimulation to
reduce congruence of bisimilarity to the search for a suitable
flexible \augmented bisimulation.
\begin{corollary}\label{cor:reduce1}
  Consider any syntactic signature $𝐝 = (Σ,(Γᵢ,dᵢ)_{i ∈ n})$.
  Let $σ$ denote the generated enhanced syntax $σ(𝐝)$.
  Let $X$ be any $σ$-transition system, and suppose that there exists a
  reflexive, \augmented, flexible bisimulation relation
  $R → X²$ such that ${∼_X^σ} ⊆ R₀$ and $R₀$ is a
  congruence. Then \augmented bisimilarity ${∼_X^σ}$ is a congruence.
\end{corollary}
\begin{proof}
  Consider any reflexive, \augmented, flexible bisimulation relation
  $R → X²$ such that $R₀$ contains \augmented
  bisimilarity and is a congruence. By
  Proposition~\ref{prop:flexbisim:bisim}, $R₀$ is \anaugmented
  bisimulation, so by terminality of $∼_X^σ$, we have $R₀ ⊆ {∼_X^σ}$,
  hence morphisms
  $$Σ₀(∼_X^σ) → Σ₀(R₀) → R₀ → {∼_X^σ}$$
  over $X₀²$.
\end{proof}

\subsection{Howe closure: basic properties}
In this section, we introduce our candidate reflexive,
\augmented, flexible bisimulation relation $R → 𝐙²$ such that ${∼_𝐙} ⊆ R₀$ and
$R₀$ is a congruence.
As is standard, we
\begin{itemize}
\item construct it directly as a congruence, 
\item prove that it is reflexive and \augmented (relatively easily),
  and, finally,
\item struggle to prove that it (or rather its transitive closure) is
  a flexible bisimulation.
\end{itemize}
\begin{definition}
  Let the \alert{Howe functor} $\Fwow∶ \psh[𝕍𝕋]/X₀² → \psh[𝕍𝕋]/X₀²$ map any $R₀ → X₀²$
  to the coproduct span $Σ₀(R₀) + (R₀;{∼_X})$, where the second term more concretely
  denotes the following composite span.
  \begin{center}
\begin{tikzcd}[ampersand replacement=\&]
	\&\& {R_0;{\sim_X}} \\
	\& {R_0} \&\& {{\sim_X}} \\
	{X_0} \&\& {X_0} \&\& {X_0}
	\arrow["{\pi_1}"', from=2-4, to=3-3]
	\arrow["{\pi_2}", from=2-4, to=3-5]
	\arrow[from=1-3, to=2-2]
	\arrow[from=1-3, to=2-4]
	\arrow["{\pi_2}", from=2-2, to=3-3]
	\arrow["{\pi_1}"', from=2-2, to=3-1]
	\arrow["\lrcorner"{anchor=center, pos=0.125, rotate=-45}, draw=none, from=1-3, to=3-3]
      \end{tikzcd}
    \end{center}
  
    Let the \alert{proof-relevant Howe closure} $R₀^⊙$ be the free
    $\Fwow$-algebra on $R₀$, and the \alert{(proof-irrelevant, or relational)
      Howe closure} $R₀^•$ denote the image of $R₀^⊙ → X₀²$.
\end{definition}
The Howe functor is a finitary endofunctor on a presheaf category,
so we have~\cite{Reiterman}:
\begin{proposition}\label{prop:wow:initialchain}
  The free $\Fwow$-algebra on any $R₀$ exists and is computed by the
  standard initial chain, and the forgetful functor
  $\Fwow\alg → \psh[𝕍𝕋]/X₀²$ is finitary monadic.
\end{proposition}

\begin{proposition}\label{prop:relwow:alt}
  Let $𝒰∶ 𝐒𝐮𝐛(X₀²) ↪ \psh[𝕍𝕋]/X₀²$ denote the canonical embedding, and
  let $Σ = \Fwow$ just for this proposition.  The composite
  endofunctor $\im ∘ Σ^* ∘ 𝒰$ on $𝐒𝐮𝐛(X₀²)$ is a monad, which is in
  fact the free monad on $\im ∘ Σ ∘ 𝒰$.  Consequently, the relational
  Howe closure $(𝒰R)^•$ on a relation $R ∈ 𝐒𝐮𝐛(X₀²)$ is the free
  $(\im ∘ Σ ∘ 𝒰)$-algebra over $R$.
\end{proposition}
\begin{proof}
  Using algebraicity of $Σ₀$, it is straightforward to show that $Σ$
  preserves epimorphisms. For any $R ∈ \psh[𝕍𝕋]/X₀²$, letting
  $T = 𝒰 ∘ \im$ denote the monad induced by the adjunction $\im ⊣ 𝒰$,
  we thus have by unique lifting a morphism $δ_R∶ ΣTR → TΣR$ as in the
  following diagram,
  \begin{center}
\begin{tikzcd}[ampersand replacement=\&]
	{\Sigma R} \&\& {𝒰 \im Σ R } \\
	{Σ  𝒰 \im R} \& {Σ (X₀²)} \& {X₀²}
	\arrow["{Σ e}"', two heads, from=1-1, to=2-1]
	\arrow["{Σ m}"', from=2-1, to=2-2]
	\arrow[from=2-2, to=2-3]
	\arrow[two heads, from=1-1, to=1-3]
	\arrow[tail, from=1-3, to=2-3]
	\arrow["{δ_R}"{pos=0.3}, dashed, from=2-1, to=1-3]
\end{tikzcd}
  \end{center}
  where $R  → X₀²$ factors as $e ∘ m$ and the last horizontal morphism is 
  $$\Fwow(X₀²) = Σ₀(X₀²) + (X₀²);∼_X \xto{[⟨a ∘ Σ₀(π₁),a ∘ Σ₀(π₂)⟩, ⟨π₁∘π₁, π₂∘π₂⟩]} X₀².$$
  The result thus follows from the next lemma.
\end{proof}

\begin{lemma}
  \label{lem:initalgtransport}
  Consider a full, reflective embedding $U∶ 𝒟 ↪ 𝒞$ from some poset $𝒟$
  into a locally finitely presentable category $𝒞$, say with left
  adjoint $L∶ 𝒞 → 𝒟$, together with a finitary endofunctor $Σ$ on $𝒞$.
  Furthermore, assume given a functor distributive law, i.e., a
  natural transformation $δ∶ Σ T → T Σ$, where $T ≔ UL$ denotes the
  induced monad.
  Then, $LΣ^*U$ is the free monad on $LΣU$, hence in particular
  the free $LΣU$-algebra on any $D ∈ 𝒟$ is $LΣ^*UD$.
\end{lemma}

\begin{lemma}\label{lem:subterm}
  In the setting of Lemma~\ref{lem:initalgtransport}, all objects of
  the form $UD ∈ 𝒞$ are \alert{subterminal}, in the sense that any two
  parallel morphisms to $UD$ are equal.
\end{lemma}
\begin{proof}
  Consider any $f,g∶ C→ UD$. By adjunction, these correspond
  bijectively to morphisms $\tilde{f},\tilde{g}∶ LC → D$, which,
  because $𝒟$ is a poset, are equal.
\end{proof}

\begin{proof}[Proof of Lemma~\ref{lem:initalgtransport}]
  By~\citet{Reiterman}, $Σ$ admits a free monad $Σ^*$.

  Furthermore, by Lemma~\ref{lem:subterm}, the given functor
  distributive law $δ$ is in fact a \alert{functor-monad distributive
    law}, in the sense that it commutes with the unit and
  multiplication of $T$.
  
  Now, by a reasoning analogous to \citet{BeckDistlaws}, functor-monad
  distributive laws $δ∶ ΣT → TΣ$ correspond bijectively to liftings of
  the monad $T$ to $Σ\alg$, i.e., monads $Tᵟ$ on $Σ\alg$ making the
  following square commute,
  \begin{center}
    \diag{%
      Σ\alg \& Σ\alg \\
      𝒞 \& 𝒞 %
    }{%
      (m-1-1) edge[labela={Tᵟ}] (m-1-2) %
      edge[labell={}] (m-2-1) %
      (m-2-1) edge[labelb={T}] (m-2-2) %
      (m-1-2) edge[labelr={}] (m-2-2) %
    }
  \end{center}
  whose multiplication and unit are mapped by the forgetful functor to
  those of $T$.  The given functor-monad distributive law $δ$ thus
  corresponds to such a lifting.  But $Σ\alg ≅ Σ^*\Alg$ over $𝒞$,
  hence we get a lifting of $T$ to $Σ^*\Alg$, which by \citet{BeckDistlaws}
  again amounts to a monad distributive law, say
  $\bar{δ}∶ Σ^*T → T Σ^*$.

  From this, using the fact that the counit is an isomorphism (which follows
  from full faithfulness of $U$), we equip the composite $L Σ^* U$ with
  monad structure:
  \begin{itemize}
  \item the unit is the composite
    $R \xto{(εᵀ)^{-1}} LUR \xto{Lη^{Σ^*}} LΣ^*UR$, 
  \item the multiplication is
    $$LΣ^*U LΣ^*U R = L Σ^* T Σ^* U R \xto{L \bar{δ}} L T Σ^*Σ^* U R
    \xto{ε U μ^{Σ^*} U R} L Σ^* U R\rlap{,}$$
  \item and the monad laws hold automatically since $𝒟$ is a poset.
  \end{itemize}

  Moreover, given any $R ∈ 𝒟$, the following are equivalent
  \begin{itemize}
  \item $Σ^*$-algebra structure (in the monad sense) on $UR$,
  \item $Σ^*$-algebra structure (in the functor sense) on $UR$,
  \item $Σ$-algebra structure on $UR$,
  \item $LΣ^*U$-algebra structure (in the monad sense) on $R$,
  \item $LΣ^*U$-algebra structure (in the functor sense) on $R$,    
  \item $LΣU$-algebra structure on $R$.
  \end{itemize}
  Indeed,
  \begin{itemize}
  \item $Σ$-algebra structure $ΣUR → UR$ corresponds by adjunction to
    $LΣU$-algebra structure $LΣUR → R$;
  \item $Σ$-algebra structure $ΣUR → UR$ corresponds by universal property of $Σ^*$
    to $Σ^*$-algebra structure $Σ^*UR → UR$ in the monad sense;
  \item by subterminality,  $Σ^*$-algebra structures $Σ^*UR → UR$
    in the monad and functor sense are equivalent;
  \item by adjunction again, $Σ^*$-algebra structure $Σ^*UR → UR$ in the functor sense
    is equivalent to $LΣ^*U$-structure  $LΣ^*UR → R$ in the functor sense; 
  \item and finally, because $𝒟$ is a poset, $LΣ^*U$-structures
    $LΣ^*UR → R$ in the functor and monad sense are equivalent.
  \end{itemize}

  We thus in particular get $(L Σ^* U)\Alg ≅ (LΣU)\alg$ over $𝒟$,
  hence the result.
\end{proof}

\begin{definition}
  Let $S^{{+};{∼_X}}$ denote the monad induced by $\Fwow$ on
  $\psh[𝕍𝕋]/X₀²$.
\end{definition}

\begin{lemma}\label{lem:wow:basicprops}
  Let $R₀'$ be the proof-relevant (resp.\ proof-irrelevant) Howe closure $R₀^⊙$ (resp.\ $R₀^•$)
  of (resp.\ a relation) $R₀$. It satisfies the following properties.
  \begin{romanenumerate}
  \item \label{item:alg} $R₀'$ is a $Σ₀$-algebra; 
  \item \label{item:act} there exists an action $R₀';{∼_X} → R₀'$
    over $X₀²$.

  Furthermore, if $X₀ = Σ₀^*(∅)$ is the initial $Σ₀$-algebra, we have:
  \item $R₀'$ is reflexive, 
  \item \label{item:wow:contains:bisim} there exists a morphism
    ${∼}_X → R₀'$ over $X₀²$.
  \end{romanenumerate}
\end{lemma}
\begin{proof} We prove the properties for the proof-relevant Howe closure -- they follow
  easily for the proof-irrelevant one.
  \begin{description}
  \item[\cref{item:alg}] By definition $R₀^⊙$ is an $\Fwow$-algebra,
    hence in particular a $Σ₀$-algebra, or more correctly an algebra
    for the obvious lifting of $Σ₀$ to $\psh[𝕍𝕋]/X₀²$.
  \item[\cref{item:act}] As an $\Fwow$-algebra, $R₀^⊙$ is an algebra for
    the second term functor, i.e., a morphism of the desired form
    $R₀^⊙;{∼_X} → R₀^⊙$.
  \end{description}
  Let us now assume that $X₀$ is the initial $Σ₀$-algebra.  Then, by
  initiality of $X₀$ and~\cref{item:alg}, there is a unique
  $Σ₀$-algebra morphism $X₀ → R₀^⊙$, which witnesses reflexivity.

  We then use reflexivity and~\cref{item:act} to construct the
  following composite
  $${∼}_X ≅ {X₀ ; {∼_X}} → {R₀^⊙;{∼_X}} → R₀^⊙\rlap{,}$$
  which proves the second point.
\end{proof}

A further crucial property is:
\begin{proposition}\label{prop:Wow:augmented}
  If $X₀$ is an $ST$-algebra, then the proof-relevant Howe closure
  $R₀^⊙$ on any $R₀$ is an $ST$-algebra, and $R₀^⊙ → X₀²$ is a
  morphism of $ST$-algebras.  Furthermore, the relational Howe closure
  $R₀^•$ is \augmented.
\end{proposition}
In order to prove this, we need a few intermediate steps.

   \begin{definition} For any bifunctor $F$ on
    a category $𝐂$ and $FΔ$-algebra $X$, let $\bar{F}$ denote the
    lifting of $F$ to $𝐂/X²$, which maps any $U → X²$ and $V → X²$ to the composite
    $$F (U,V) → F(X²,X²) → F(X,X)² → X².$$
  \end{definition}

  \begin{lemma}\label{lem:wdistrib}
    For any bifunctor $Γ$ on a category $𝐂$ with pullbacks, object
    $X ∈ 𝐂$, and spans $uᵢ∶ Uᵢ → X²$, for $i ∈ 3$, there is a
    morphism $$Γ ((U₁;U₂),U₃) → Γ (U₁,U₃) ; Γ (U₂,X)$$
    of spans over $X$.
  \end{lemma}
  \begin{proof}
    We construct the desired morphism by universal property of
    pullback, as in the following diagram.
    \begin{center}
\begin{tikzcd}[ampersand replacement=\&]
	{\Gamma_{}((U₁;{U₂}),U₃)} \&\& {\Gamma_{}({U₂},U₃)} \\
	\& {\Gamma_{}(U₁,U₃);\Gamma_{}({U₂},X)} \&\& {\Gamma_{}({U₂},X)} \\
	{\Gamma_{}(U₁,U₃)} \&\& {\Gamma_{}(X,U₃)} \\
	\& {\Gamma_{}(U₁,U₃)} \&\& {\Gamma_{}(X,X)}
	\arrow[from=1-1, to=3-1]
	\arrow[from=1-1, to=1-3]
	\arrow["{\Gamma_{}(\pi_1,U₃)}"'{pos=0.2}, from=1-3, to=3-3]
	\arrow["{\Gamma_{}(\pi_2,U₃)}"'{pos=0.3}, from=3-1, to=3-3]
	\arrow[""{name=0, anchor=center, inner sep=0, pos=.15}, "{\pi_1}"{pos=0.8}, fore, from=2-2, to=4-2]
	\arrow[""{name=1, anchor=center, inner sep=0, pos=-.2}, "{\pi_2}"'{pos=0.2}, fore, from=2-2, to=2-4]
	\arrow["{\Gamma_{}(\pi_1,X)}", fore, from=2-4, to=4-4]
	\arrow["{\Gamma_{}(\pi_2,\pi_2)}"', from=4-2, to=4-4]
	\arrow[Rightarrow, no head, from=3-1, to=4-2]
	\arrow["{\Gamma_{}(X,\pi_2)}"'{pos=0}, from=3-3, to=4-4]
	\arrow["{\Gamma_{}({U₂},\pi_2)}"{pos=0.9}, from=1-3, to=2-4]
	\arrow[dashed, from=1-1, to=2-2]
	\arrow[shorten <=6pt, shorten >=6pt, no head, from=0, to=1,
        to path={-| (\tikztotarget)}]        
\end{tikzcd}
    \end{center}    
  \end{proof}

  \begin{lemma}
    Assume that $X₀$ is a $σ$-algebra with structure given by
    \begin{mathpar}
      𝐚∶ Σ₀X₀ → X₀ \and … \and 𝐛ᵢ∶ Γᵢ(X₀,X₀) → X₀ \and … \rlap{ ,}
    \end{mathpar}
    and let the derived monad algebra structures be as follows.
    \begin{mathpar}
      \bar{𝐚}∶ SX₀ → X₀ \and … \and \bar{𝐛}_{<i}∶ TᵢX₀ → X₀ \and …\rlap{.}
    \end{mathpar}

    Then, for all $i ∈ n$, the incremental structural law
     $$dᵢ∶ Γᵢ(Σ₀A,B) → STᵢ (Γᵢ(A,STᵢB)+ A + B)$$
     lifts to an incremental structural law
     $$\bar{dᵢ}∶ \bar{Γᵢ}(Σ₀^{{+};{∼_X}}A,B) → S^{{+};{∼_X}}\bar{Tᵢ}
     (\bar{Γᵢ}(A,S^{{+};{∼_X}}\bar{Tᵢ}B) + A + B)\rlap{.}$$ 
   \end{lemma}
   \begin{proof}
     By Lemma~\ref{lem:wdistrib}, using left-cocontinuity of $Γᵢ$, and
     the fact that $∼_X$ is \augmented.
   \end{proof}

   \begin{proof}[Proof of Proposition~\ref{prop:Wow:augmented}]
     By Proposition~\ref{prop:sigmad}, there exists a distributive law
     $$\bar{T}_{n+1}S^{{+};{∼_X}} → S^{{+};{∼_X}}\bar{T}_{n+1}$$ and
     $\bar{T}_{n+1}$ is constant-free, hence the natural
     transformation $S^{{+};{∼_X}} → S^{{+};{∼_X}} \bar{T}_{n+1}$ is
     an isomorphism at $∅$.  The proof-relevant Howe closure
     $R₀^⊙ = S^{{+};{∼_X}}∅$ thus acquires a canonical
     $S^{{+};{∼_X}}\bar{T}_{n+1}$-algebra structure. The terminal
     object also is one, of course, and the unique morphism to it is a
     $S^{{+};{∼_X}}\bar{T}_{n+1}$-algebra morphism, which completes
     the proof of the first point.
     
     The proof-relevant Howe closure is in particular \augmented via
     $$Γᵢ(R₀^⊙,X) →Γᵢ(R₀^⊙,R₀^⊙) → R₀^⊙\rlap{,}$$
     which entails \augmentedness for the relational Howe closure by
     the fact that each $Γᵢ$, being left-cocontinuous, preserves
     epimorphisms in its first argument, and that all epimorphisms are
     strong in presheaf categories.  Indeed, we find the desired
     morphism by lifting as in the following diagram.
     \begin{center}
       \begin{tikzcd}[ampersand replacement=\&]
         {\Gamma(R_0^\odot,X_0)} \& {\Gamma(R_0^\odot,R_0^\odot)} \&\& {R_0^\odot} \\
         {\Gamma(R_0^\bullet,X_0)} \&\&\& {R_0^\bullet} \\
         {\Gamma(X_0^2,X_0)} \& {\Gamma(X_0^2,X_0^2)} \& {\Gamma(X_0,X_0)^2} \& {X_0^2}
         \arrow[two heads, from=1-1, to=2-1]
         \arrow[from=2-1, to=3-1]
         \arrow[from=3-3, to=3-4]
         \arrow[two heads, from=1-4, to=2-4]
         \arrow[hook, from=2-4, to=3-4]
         \arrow[from=1-1, to=1-2]
         \arrow[from=1-2, to=1-4]
         \arrow[from=1-2, to=3-2]
         \arrow[from=3-2, to=3-3]
         \arrow[from=3-1, to=3-2]
         \arrow[dashed, hook, from=2-1, to=2-4]
       \end{tikzcd}
     \end{center}
   \end{proof}

   A final basic property is about symmetry of the relational transitive closure of the relational Howe closure
   on the syntactic transition system (Proposition~\ref{prop:trans:is:sym} below). \tom{work with relations from earlier on?}
   \begin{definition}[{\cite[Definition~9.5]{HL}}]\label{def:relational:transitive:closure}
     The \alert{relational transitive closure} $R₀^{\bar{+}}$ of a
     span $R₀ → X₀²$ is the union $⋃_{n >0} \im(R₀^{;n})$, where $(-)^{;n}$
     denotes iterated self-composition of spans.
   \end{definition}

   \begin{proposition}\label{prop:trans:action}
     For any span $R₀ → X₀²$, the relational transitive closure
     $R₀^{\bar{+}}$ is equipped with an action
     $R₀;R₀^{\bar{+}} → R₀^{\bar{+}}$ over $X₀²$.
   \end{proposition}
   The proof relies on the following lemma.
   \begin{lemma}\label{lem:seqcomp:cocont}
     In any complete, cocomplete, regular, and locally cartesian closed
     category, hence in particular in any presheaf category,
     \begin{romanenumerate}
     \item \label{item:spancomp:cocont} span composition preserves all
       colimits, on both sides, and
     \item \label{item:seqcomp:cocont} sequential composition of
       relations preserves all unions, on both sides.
     \end{romanenumerate}
   \end{lemma}
   \begin{proof}
     The pullback functor (along the relevant projection), being a left
     adjoint, is cocontinuous, which directly entails the first point.
     For the second point, in a regular category, the pullback functor
     preserves regular epis and monos, hence image factorisations.
   \end{proof}
   \begin{proof}[Proof of Proposition~\ref{prop:trans:action}]
     We have $$R₀ ; ⋃_{n>0} \im(R₀^{;n}) ↠ \im(R₀) ; ⋃_{n>0} \im(R₀^{;n}) ≅ ⋃_{n>1} \im(R₀^{;n}) ↪ ⋃_{n>0} \im(R₀^{;n})\rlap{,}$$
     where the isomorphism holds by Lemma~\ref{lem:seqcomp:cocont}\cref{item:seqcomp:cocont}.
   \end{proof}

   \begin{proposition}\label{prop:trans:is:sym}
     Let again $X₀ = Σ₀^*(∅)$. Then the relational transitive closure
     $∅^{•\bar{+}}$ of the proof-irrelevant Howe closure of $∅$ is
     symmetric.
   \end{proposition}

   \begin{lemma}[{\cite[Lemma~9.10]{HL}}]
     For any span $R₀ → 𝐙₀²$, if there exists a span morphism
     $R₀ → R₀^{\bar{+}†}$, then $R₀^{\bar{+}}$ is symmetric.
   \end{lemma}

   \begin{lemma}
     If a span $R$ is symmetric, in the sense that there is a morphism
     $R^† → R$ over $X₀²$, then so is its induced relation.
   \end{lemma}
   \begin{proof}
     We proceed as in the following diagram.
     \begin{center}
       \hfill
\begin{tikzcd}[ampersand replacement=\&,baseline=(\tikzcdmatrixname-4-5.base)]
	R \&\&\&\&\&\& R \\
	{im R} \&\& {X²} \&\&\&\& {im R} \\
	\\
	\&\&\&\& {X^2}
	\arrow["s", from=1-1, to=1-7]
	\arrow["{⟨ π₁, π₂ ⟩}"'{pos=0.7}, from=1-7, to=4-5]
	\arrow["e", two heads, from=1-7, to=2-7]
	\arrow["m", hook, from=2-7, to=4-5]
	\arrow["{⟨π₂,π₁⟩}"{pos=0.6}, from=2-3, to=4-5]
	\arrow["{⟨ π₁, π₂ ⟩ }"{pos=0.7}, from=1-1, to=2-3]
	\arrow["e", two heads, from=1-1, to=2-1]
	\arrow["m", hook, from=2-1, to=2-3]
	\arrow["{⟨π₂,π₁⟩ ∘ m}"', hook, from=2-1, to=4-5]
	\arrow[curve={height=12pt}, dashed, from=2-1, to=2-7]
\end{tikzcd}
\qedhere
     \end{center}
   \end{proof}

   \begin{proof}[Proof of Proposition~\ref{prop:trans:is:sym}]
     By the lemma, it suffices to construct a morphism
     $∅^• → ∅^{•\bar{+}†}$.  Thus, by
     Proposition~\ref{prop:relwow:alt}, it suffices to endow
     $∅^{•\bar{+}†}$ with algebra structure for the endofunctor
     $S ↦ \im(\Fwow(S))$ on $𝐒𝐮𝐛(X₀²)$.
     For this, because $\Fwow$ is algebraic, it suffices to endow 
     $∅^{•\bar{+}†}$ with $\Fwow$-algebra structure.

     We first equip it with $({-};{∼_X})$-algebra structure.  We need
     to find a morphism $∅^{•\bar{+}†};{∼_X} → ∅^{•\bar{+}†}$ over
     $X₀²$, or equivalently by applying the involution $(-)^†$, a
     morphism ${∼_X}^†;∅^{•\bar{+}} → ∅^{•\bar{+}}$. We pick the
     composite
  $${∼_X}^†;∅^{•\bar{+}} → {∼_X};∅^{•\bar{+}} → ∅^{•};∅^{•\bar{+}}  → ∅^{•\bar{+}}\rlap{,}$$
  where
  \begin{itemize}
  \item the first morphism is symmetry of $∼_X$, 
  \item the second morphism is that of Lemma~\ref{lem:wow:basicprops}, 
  \item the last morphism is the action from Proposition~\ref{prop:trans:action}.
  \end{itemize}

  This leaves us with the task of equipping $∅^{•\bar{+}†}$ with
  algebra structure for the lifting of $Σ₀$ to $𝐒𝐮𝐛(X₀²)$, for which
  it suffices, by algebraicity of $Σ₀$, to equip it with algebra
  structure for the lifting of $Σ₀$ to $\psh[𝕍𝕋]/X₀²$, say
  $\bar{Σ}₀$. By \citet[Corollary~9.8]{HL}, we have
  $∅^{•\bar{+}†} ≅ ∅^{•†\bar{+}}$, and by \citet[Lemma~9.9]{HL},
  $∅^{•†\bar{+}}$ is the colimit of the chain
  \[
    X₀ → \im (∅^{•†}) ≅ \im (∅^{•†};X₀) → \im (∅^{•†};∅^{•†}) ≅ \im (∅^{•†};∅^{•†};X₀) → \im (∅^{•†};∅^{•†};∅^{•†}) → {…}
  \]
  in $\psh[𝕍𝕋]/X₀²$.  But $\bar{Σ}₀$ is algebraic, hence the forgetful
  functor $\bar{Σ}₀\alg → \psh[𝕍𝕋]/X₀²$ creates filtered colimits,
  hence in particular colimits of chains.  It thus suffices to lift
  the above chain to $\bar{Σ}₀\alg$.  Furthermore, because all objects
  of the chain are relations, they are subterminal, hence all
  morphisms will automatically lift to $\bar{Σ}₀\alg$ if the objects
  do.  Finally, the forgetful functor $\bar{Σ}₀\alg → \psh[𝕍𝕋]/X₀²$
  creates limits, and $∅^{•}$ possesses $Σ₀$-algebra structure by
  Lemma~\ref{lem:wow:basicprops}, hence so does $∅^{•†}$.  \end{proof}

   To conclude this section, we use the basic facts we just proved to
   reduce the main result to the fact that $∅^•$ is a flexible
   simulation.

   \begin{proposition}\label{prop:reduce2}
     Consider any syntactic signature $𝐝 = (Σ,(Γᵢ,dᵢ)_{i ∈ n})$, and
     suppose that $∅^•_𝐙$ is a flexible simulation.  Then \augmented
     bisimilarity on $𝐙$ is a congruence.
   \end{proposition}
   \begin{remark}
     Let us recall that by Definition~\ref{def:flexible}, $R₀ → X₀²$
     in $\psh[𝕍𝕋]$ is a flexible simulation when its cartesian lifting
     $R₀^{⇑} → X²$ is.
   \end{remark}
   We will rely on the following lemmas.
\begin{lemma}\label{lem:colimitcreation:cancellation}
  Consider any commutative diagram of functors between locally small categories
  \begin{center}
    \diag{%
      𝐀 \& \& 𝐁 \\
      \& 𝐂 %
    }{%

      (m-1-1) edge[labela={U}] (m-1-3) %
      edge[labelbl={V}] (m-2-2) %
      (m-1-3) edge[labelbr={W}] (m-2-2) %
    }
  \end{center}
  If $V$ and $W$ create colimits of a certain shape $D$,
  and $W$ preserves them (typically if $𝐂$ has them), then $U$ creates
  them.
\end{lemma}
\begin{proof}
  Consider any functor $J∶ D → 𝐀$ and colimiting cocone $K∶ D^⊤ → 𝐁$
  for $U∘J$.  Because $W$ preserves colimits of shape $D$, $W∘K$ is
  colimiting for $W∘U∘J$, hence because $V$ creates colimits, we find
  a unique lifting $J^↑$ such that $J^↑∘I = J$ and $V ∘ J^↑ = L ≔ W∘K$, as in the following diagram.
  \begin{center}
    \begin{tikzcd}[ampersand replacement=\&]
      D \&\& {𝐀} \&\& {𝐁} \\
      \\
      {D^⊤} \&\& {𝐂} \&\& {𝐂}
      \arrow["U", from=1-3, to=1-5]
      \arrow["I"', hook, from=1-1, to=3-1]
      \arrow["J", from=1-1, to=1-3]
      \arrow["K"'{pos=0.3}, curve={height=6pt}, from=3-1, to=1-5]
      \arrow["L"', from=3-1, to=3-3]
      \arrow["{J^↑}", dashed, from=3-1, to=1-3]
      \arrow["V"{pos=0.3}, from=1-3, to=3-3]
      \arrow["W", from=1-5, to=3-5]
      \arrow[Rightarrow, no head, from=3-3, to=3-5]
    \end{tikzcd}
  \end{center}
  But now $U ∘ J^↑$ and $K$ both are candidate liftings for the outer
  rectangle, so by uniqueness in the creation of colimits by $W$ they
  are equal, and thus $J^↑$ is a lifting for the original square
  $(J,K)$.

  Furthermore, any lifting for $(J,K)$ induces one for $(J,L)$, hence
  should be equal to $J^↑$, which proves uniqueness.

  Finally, $J^↑$ is colimiting because $V$ creates colimits of shape
  $D$.
\end{proof}

\begin{definition}
  Given a bifunctor $Γ∶ 𝐂² → 𝐂$ and an object $X$, a
  \alert{$(Γ,X)$-premodule} is an object $M$ equipped with an
  \alert{action}, i.e., a morphism $r∶ Γ(M,X) → M$.  A morphism of
  $(Γ,X)$-premodules is a morphism commuting with action.  We let $(Γ,X)\Mod$
  denote the category of $(Γ,X)$-premodules.
\end{definition}
\begin{terminology}
  When $Γ$ is clear from context, we often omit it and talk about
  $X$-premodules and $X\Mod$.
\end{terminology}
\begin{remark}
 An enhanced span $R→X²$ as in Definition~\ref{def:enhanced-span} is a span in the category of $X$-premodules.
\end{remark}
\begin{lemma}\label{lem:Xmod}
  If $Γ$ is left-cocontinuous and $𝐂$ is locally finitely presentable
  and regular, then the category $X\Mod$ is regular and the forgetful
  functor $X\Mod → 𝐂$ creates all limits and colimits, as well as
  image factorisations.
\end{lemma}
\begin{proof}
  Creation of limits and colimits follows easily from the fact that
  $X\Mod$ is the category of algebras for the cocontinuous endofunctor
  $Γ({-},X)$.

  In particular, $X\Mod$ is complete and cocomplete, hence regularity
  reduces to showing that regular epis are stable under pullback.
  
  Let us first appeal to~\cite[§1.6.5]{HL} for definitions and
  preliminary results about images. Notably, in a locally finitely
  presentable category, (strong epi)-mono factorisations yield image
  factorisations, and union may be computed by cotupling followed by
  (strong epi)-mono factorisation.

  Let us then consider any pullback square
  \begin{center}
    \Diag{%
      \stdpbk %
    }{%
      A \& B \\
      C \& D
    }{%
      (m-1-1) edge[labela={u}] (m-1-2) %
      edge[labell={v}] (m-2-1) %
      (m-2-1) edge[labelb={g}] (m-2-2) %
      (m-1-2) edge[labelr={f}] (m-2-2) %
    }
  \end{center}
  in $X\Mod$, with $f$ a regular epi, and show that $v$ must also be a
  regular epi.  By creation, hence preservation, of limits and
  colimits, the given pullback square is also a pullback in $𝐂$ and
  $f$ is a regular epi there too. So by regularity of $𝐂$, $v$ is a
  regular epi in $𝐂$. Equivalently, it is a coequaliser of its kernel
  pair. But by creation of limits the kernel pair uniquely lifts to a
  kernel pair in $X\Mod$, and by creation of colimits $v$ is a
  coequaliser there too. This shows that $X\Mod$ is regular.


  Finally, let us prove that the forgetful functor creates image
  factorisations.  Given $A,C ∈ X\Mod$, let us consider any image
  factorisation $A \xarrow[onto]{e} B \xarrow[into]{m} C$ in $𝐂$ of a
  morphism $f∶ A → C$ in $X\Mod$, i.e., $e$ is a regular epi and $m$
  is a mono in $𝐂$.  In this situation, $e$ is the coequaliser of its
  kernel pair in $𝐂$, but, as we just saw, this kernel pair lifts to a
  kernel pair in $X\Mod$, whose coequaliser is created by the
  forgetful functor, hence $e$ is a coequaliser, hence a regular epi
  in $X\Mod$. Finally, $f$ also coequalises the kernel pair, hence the
  existence of a unique mediating morphism $B → C$ in $X\Mod$, which
  must be $m$ by faithfulness of the forgetful functor $X\Mod →
  𝐂$. Thus, $m$ is also a morphism in $X\Mod$. Finally, its monicity
  follows again by faithfulness of the forgetful functor.
\end{proof}
\begin{lemma}\label{lem:augmented:images}
  If $𝐂$ is regular, then \augmented spans are stable under
  images, that is if $p∶ R → X²$ is \augmented, then so is
  $\im(p)∶ \im(R) ↪ X²$.
\end{lemma}
\begin{proof}
  By Lemma~\ref{lem:Xmod} (creation of image factorisations).
\end{proof}

\begin{lemma}\label{lem:augmented:coprod}
  The forgetful functor $X\Mod/X² → 𝐂/X²$ creates colimits. Hence, in
  particular (by cocompleteness of $𝐂/X²$), \augmented spans are
  closed under all colimits in $𝐂/X²$.
\end{lemma}
\begin{proof}
  Consider the following commutative diagram in $𝐂𝐀𝐓$.
  \begin{center}
    \begin{tikzcd}[ampersand replacement=\&]
      {X\Mod/X²} \&\& {𝐂/X²} \\
      X\Mod \&\& {𝐂}
      \arrow[from=1-1, to=1-3]
      \arrow[from=1-1, to=2-1]
      \arrow[from=2-1, to=2-3]
      \arrow[from=1-3, to=2-3]
    \end{tikzcd}    
  \end{center}
  Colimits are created by both (vertical) projection functors, and
  also by the bottom functor by Lemma~\ref{lem:Xmod}.
  Furthermore, $𝐂$ being cocomplete, the projection functor $𝐂/X² → 𝐂$
  preserves all colimits, hence by
  Lemma~\ref{lem:colimitcreation:cancellation} the top functor
  creates them.
\end{proof}

   \begin{lemma}\label{lem:rafss}
     For any syntactic signature $𝐝 = (Σ,(Γᵢ,dᵢ)_{i ∈ n})$ and
     $X ∈ σ(𝐝)\Trans$, if $R₀ ↪ X₀²$ in $\psh[𝕍𝕋]$ is a reflexive,
     \augmented flexible simulation relation, then so is
     $R₀^{\bar{+}}$.
   \end{lemma}
   \begin{proof}
     Reflexivity is clear. For \augmentedness, we have seen in
     Lemmas~\ref{lem:augmented:images} and~\ref{lem:augmented:coprod}
     that \augmented spans are closed under images and coproducts.
     Furthermore, closedness under span composition follows directly
     by Lemma~\ref{lem:wdistrib}.  Finally, in order to show that
     $R₀^{\bar{+}}$ is a flexible simulation, we adopt the
     characterisation of \citet[Lemma~9.9]{HL}, by which
     $R₀^{\bar{+}}$ is the colimit of the chain
     \[
       X₀ → \im (R₀) ≅ \im (R₀;X₀) → \im (R₀;R₀) ≅ \im (R₀;R₀;X₀) → \im (R₀;R₀;R₀) → {…}
     \]
     in $\psh[𝕍𝕋]/X₀²$.  By Corollary~\ref{cor:filtered:flex:filtered}, it suffices
     to show that each $\im (R₀^{{;}n})$ is a flexible simulation.  By
     Lemma~\ref{lem:images:flex}, it further suffices to show that
     each $R₀^{{;}n}$ is a flexible simulation.  By induction and
     Lemma~\ref{lem:spancomp:flex}, it finally suffices to show that
     $R₀$ is a flexible simulation, which it is by hypothesis.
   \end{proof}
   
   \begin{proof}[Proof of Proposition~\ref{prop:reduce2}]
     By hypothesis $∅^•$ is a flexible simulation. It is also
     \augmented by Proposition~\ref{prop:Wow:augmented}. Let now
     $R₀ ≔ ∅^{•\bar{+}}$, which is again a flexible \augmented
     simulation by Lemma~\ref{lem:rafss}.
     By Proposition~\ref{prop:trans:is:sym}, $R₀$ is moreover
     symmetric. But any symmetric simulation is in fact a
     bisimulation, so $R₀$ is a flexible \augmented bisimulation.
     Furthermore, $R₀$ contains $∼_𝐙$ by
     Lemma~\ref{lem:wow:basicprops}\cref{item:wow:contains:bisim}, and
     is a congruence by Lemma~\ref{lem:wow:basicprops}\cref{item:alg}.
     We thus conclude by Corollary~\ref{cor:reduce1}.
   \end{proof}
   
\subsection{The key lemma}
We at last introduce the key lemma, which will directly lead us to a
proof of Theorem~\ref{thm:main}.

\begin{lemma}\label{lem:key}
  For any syntactic signature $𝐝 = (Σ,(Γᵢ,dᵢ)_{i ∈ n})$, if $Σ₁$
  preserves functional flexible bisimulations, then the cartesian
  lifting $∅^{•⇑}_𝐙$ of $∅^•_𝐙$ is a flexible simulation.
\end{lemma}

Before proving the lemma, let us prove the main theorem,
as promised:
\begin{proof}[Proof of Theorem~\ref{thm:main}]  By
  Proposition~\ref{prop:reduce2}, it suffices to prove that $∅^•_𝐙$ is a
  flexible simulation, which is the case by Lemma~\ref{lem:key}.
\end{proof}

The rest of this section is a proof of Lemma~\ref{lem:key}.

\begin{notation}
  We abbreviate $∅^⊙_𝐙$ to $∅^⊙$ and $∼^{σ(𝐝)}_𝐙$ to $∼$.
\end{notation}

In order to prove that $∅^•_𝐙$ is a simulation, it suffices
to prove that $∅^⊙_𝐙$ is, by Lemma~\ref{lem:wpbk}.

Briefly, we will construct an $ω$-chain of flexible simulations of the
form
$$\check{Σ}₁ⁿ(𝐙₀) ← Rⁿ → 𝐙\rlap{,}$$
whose projection to $\psh[𝕍𝕋]$ is the constant chain on
\begin{equation}
𝐙₀ ← ∅^⊙ → 𝐙₀\rlap{.}\label{eq:Wowspan}
\end{equation}
By construction, the colimit of this chain will be a flexible
simulation
$$𝐙 ← R^∞ → 𝐙$$
with projection
$$𝐙₀ ← ∅^⊙ → 𝐙₀\rlap{,}$$
which entails by Lemma~\ref{lem:bisim:epi} that $∅^⊙$ is a flexible
simulation as desired.

For this, let us construct a category whose objects are spans of a similar form.
\begin{definition}
  Let $𝐒𝐩𝐚𝐧/∅^⊙$ denote the limit of the diagram
  $$\psh[𝔼𝕋]/Δ𝐙₀ \xot{\psh[𝔼𝕋]/Δπ₁} \psh[𝔼𝕋]/Δ∅^⊙ \xto{\psh[𝔼𝕋]/Δπ₂} \psh[𝔼𝕋]/Δ𝐙₀ \xot{𝐙} 1$$
  weighted by
    $$𝟚 \xot{0} 1 \xto{0} 𝟚 \xot{1} 1$$
\end{definition}
\begin{remark}
  A weighted cone from some category $A$ is thus a diagram of the form
  \begin{center}
\begin{tikzcd}[ampersand replacement=\&]
	\&\& A \&\& 1 \\
	\\
	{\psh[𝔼𝕋]/Δ𝐙₀} \&\& {\psh[𝔼𝕋]/Δ∅^⊙} \&\& {\psh[𝔼𝕋]/Δ𝐙₀}
	\arrow["{\psh[𝔼𝕋]/Δπ₁}", from=3-3, to=3-1]
	\arrow["{\psh[𝔼𝕋]/Δπ₂}"', from=3-3, to=3-5]
	\arrow[""{name=0, anchor=center, inner sep=0}, "{𝐙}", from=1-5, to=3-5]
	\arrow[""{name=1, anchor=center, inner sep=0}, "Y"', from=1-3, to=3-1]
	\arrow[""{name=2, anchor=center, inner sep=0}, "X"{pos=0.8}, from=1-3, to=3-3]
	\arrow["{!}", from=1-3, to=1-5]
	\arrow[shorten <=8pt, shorten >=8pt, Rightarrow, from=2, to=1]
	\arrow[shorten <=16pt, shorten >=16pt, Rightarrow, from=2, to=0]
\end{tikzcd}
  \end{center}

  Hence, objects of the weighted limit are spans of the form $Y ← X → 𝐙$
  over~\eqref{eq:Wowspan}, and a morphism from such a span to some
  span $Y' ← X' → 𝐙$ is a pair $(g∶ Y₁ → Y'₁, f∶ X₁ → X'₁)$ of
  morphisms in $\psh[𝔼𝕋]$ making the following diagram commute.
  \begin{center}
\begin{tikzcd}[ampersand replacement=\&]
	{Y₁} \&\& {X_1} \&\& {𝐙₁} \\
	\& {Y'₁} \&\& {X'₁} \\
	{Δ𝐙₀} \&\& {Δ∅^⊙} \&\& {Δ𝐙₀}
	\arrow[from=1-3, to=1-1]
	\arrow["f", from=1-3, to=2-4]
	\arrow["g", from=1-1, to=2-2]
	\arrow[from=1-1, to=3-1]
	\arrow[from=2-2, to=3-1]
	\arrow[from=1-3, to=3-3]
	\arrow[from=2-4, to=3-3]
	\arrow[from=3-3, to=3-1]
	\arrow[from=1-3, to=1-5]
	\arrow[from=2-4, to=1-5]
	\arrow[from=1-5, to=3-5]
	\arrow[from=3-3, to=3-5]
	\arrow[from=2-4, to=2-2, fore]
\end{tikzcd}    
  \end{center}
\end{remark}

\begin{proposition}
  The forgetful functor to $\psh[𝔼𝕋]²$ mapping any span
  $Y ← X → 𝐙$ to $(Y₁,X₁)$ creates colimits and connected limits.
\end{proposition}
\begin{proof}
  Straightforward.
\end{proof}

\begin{proposition}
  The category $𝐒𝐩𝐚𝐧/∅^⊙$ has as initial object the span
  $𝐙₀ ← ∅^⊙ → 𝐙$.
\end{proposition}

Since $𝐙₀ = \check{Σ}₁ⁿ(𝐙₀)$, this span has the desired form, and its
left-hand leg is trivially a functional flexible bisimulation, so we
may take it as our $R⁰$.

\begin{definition}
  Let $F$ denote the endofunctor on $𝐒𝐩𝐚𝐧/∅^⊙$ that maps any object
  \begin{center}
    \diag{%
      Y₁ \& X₁ \& 𝐙₁ \\
      Δ𝐙₀ \& Δ∅^⊙ \& 𝐙₀ %
    }{%
      (m-1-1) edge[<-,labela={}] (m-1-2) %
      edge[labell={}] (m-2-1) %
      (m-2-1) edge[<-,labelb={}] (m-2-2) %
      (m-1-2) edge[labela={}] (m-1-3) %
      edge[labell={}] (m-2-2) %
      (m-2-2) edge[labelb={}] (m-2-3) %
      (m-1-3) edge[labelr={}] (m-2-3) %
    }
  \end{center}
to
  \begin{center}
    \diag{%
      Σ₁(Y)₁ \& Σ₁(X)₁;(∼^*)^↑ \& 𝐙₁ \\
      Δ_{𝟚,𝐬,𝐥,𝐭}(Σ₀^?(𝐙₀),𝐙₀) \& Δ_𝐬(Σ₀^?(∅^⊙);∼^*) × Δ_𝐥(∅^⊙;𝐙₀) × Δ_𝐭(∅^⊙;∼^*) \& Δ_{𝟚,𝐬,𝐥,𝐭}(Σ₀^?(𝐙₀),𝐙₀) \\
      Δ𝐙₀ \& Δ∅^⊙ \& 𝐙₀. %
    }{%
      (m-1-1) edge[<-,labela={}] (m-1-2) %
      edge[labell={}] (m-2-1) %
      (m-2-1) edge[<-,labelb={}] (m-2-2) %
      (m-1-2) edge[labela={}] (m-1-3) %
      edge[labell={}] (m-2-2) %
      (m-2-2) edge[labelb={}] (m-2-3) %
      (m-1-3) edge[labelr={}] (m-2-3) %
      (m-2-1)
      edge[labell={}] (m-3-1) %
      (m-3-1) edge[<-,labelb={}] (m-3-2) %
      (m-2-2)
      edge[labell={}] (m-3-2) %
      (m-3-2) edge[labelb={}] (m-3-3) %
      (m-2-3) edge[labelr={}] (m-3-3) %

    }
  \end{center}
\end{definition}

\begin{definition}
Let $Rⁿ$ be the initial $F$-chain.
\end{definition}

\begin{lemma}
  The endofunctor $F$ preserves flexible simulations.
\end{lemma}
\begin{proof}
  Given any span $Y ← X → 𝐙$ over~\eqref{eq:Wowspan}, the left leg of
  its image under $F$ is functional flexible simulation iff the following pasting is a pointwise weak pullback.
  \begin{center}
\begin{tikzcd}[ampersand replacement=\&]
	{Σ₁(X)₁;(∼^*)^↑} \& {Σ₁(X)₁} \& {Σ₁(Y)_1} \\
	\\
	{Δ_𝐬(Σ₀^?(∅^⊙);∼^*) × Δ_𝐥(∅^⊙;𝐙₀)} \& {Δ_{\mathbb{2},𝐬,𝐥}(Σ₀^?(∅^⊙),∅^⊙)} \& {Δ_{\mathbb{2},𝐬,𝐥}(Σ₀^?(𝐙₀),𝐙₀)} \\
	{Δ_{𝐬,𝐥}∅^⊙} \&\& {Δ_{𝐬,𝐥}𝐙₀}
	\arrow[from=4-1, to=4-3]
	\arrow[from=3-1, to=4-1]
	\arrow[from=3-3, to=4-3]
	\arrow[from=1-3, to=3-3]
	\arrow[from=3-1, to=3-2]
	\arrow[from=3-2, to=3-3]
	\arrow[from=1-2, to=3-2]
	\arrow[from=1-2, to=1-3]
	\arrow[from=1-1, to=3-1]
	\arrow[from=1-1, to=1-2]
\end{tikzcd}
  \end{center}

  For this, by \citet[Lemma~9.26,(i)]{HL}, it suffices to prove that
  all three inner polygons are pointwise weak pullbacks.  The top
  right square is a pointwise weak pullback because $Σ₁$ preserves
  functional flexible bisimulations.  The top left square also is a
  pointwise weak pullback, as the left face of the following cube.
    \begin{center}
\begin{tikzcd}[ampersand replacement=\&]
	{Σ₁(X)₁;(∼^*)^↑} \&\& {(∼^*)^↑} \\
        \& {Σ₁(X)₁} \&\& {𝐙₁} \\
	{Δ_𝐬(Σ₀^?(∅^⊙);∼^*) × Δ_𝐥(∅^⊙;𝐙₀)} \&\& {Δ_𝐬(∼^*)×Δ_𝐥(𝐙₀)} \\
	\& {Δ_{\mathbb{2},𝐬,𝐥}(Σ₀^?(∅^⊙),∅^⊙)} \&\& {Δ_{𝐬,𝐥}(𝐙₀)}
	\arrow[from=3-1, to=4-2]
	\arrow[from=2-2, to=4-2]
	\arrow[from=1-1, to=2-2]
	\arrow[from=2-2, to=2-4]
	\arrow[from=4-2, to=4-4]
	\arrow[from=2-4, to=4-4]
	\arrow[from=1-3, to=2-4]
	\arrow[from=1-3, to=3-3]
	\arrow[from=3-3, to=4-4]
	\arrow[from=1-1, to=1-3]
	\arrow[from=3-1, to=3-3]
	\arrow[from=1-1, to=3-1]
\end{tikzcd}
    \end{center}
    Indeed, the top and bottom faces are pullbacks by construction,
    and the right face is a pointwise weak pullback because $∼^*$ is a
    bisimulation. The left face being a pointwise weak pullback thus
    follows by \citet[Lemma~9.26,(i)]{HL}.

    Finally, for the bottom rectangle, its domain is 
    $$
    \begin{array}{rcl}
      Δ_𝐬(Σ₀^?(∅^⊙);∼^*) × Δ_𝐥(∅^⊙;𝐙₀)             & ≅ & Δ_𝐬(Σ₀^?(∅^⊙);∼^*) × Δ_𝐥(∅^⊙) \\
                                     & ≅ & Δ_𝐬(Σ₀(∅^⊙);∼^* + ∅^⊙;∼^*) × Δ_𝐥(∅^⊙) \\
                                     & ≅ & (Δ_𝐬(Σ₀(∅^⊙);∼^*) + Δ_𝐬(∅^⊙;∼^*)) × Δ_𝐥(∅^⊙) \\
                                     & ≅ & Δ_𝐬(Σ₀(∅^⊙);∼^*) × Δ_𝐥(∅^⊙) + Δ_𝐬(∅^⊙;∼^*)× Δ_𝐥(∅^⊙) \\
                                     & ≅ & Δ_{𝟚,𝐬,𝐥}((Σ₀(∅^⊙);∼^*),∅^⊙) +     Δ_{𝟚,𝐬,𝐥}((∅^⊙;∼^*),∅^⊙).
    \end{array}$$
    Similarly, we have
    $Δ_{𝟚,𝐬,𝐥}(Σ₀^?𝐙₀,𝐙₀) ≅ Δ_{𝟚,𝐬,𝐥}(Σ₀𝐙₀,𝐙₀) + Δ_{𝟚,𝐬,𝐥}(Σ₀^?𝐙₀,𝐙₀)$.
    The rectangle is thus obtained by applying $Δ_{𝟚,𝐬,𝐥}$ to the (vertical) copairing of
    the following two squares.
    \begin{center}
\begin{tikzcd}[ampersand replacement=\&]
	{Σ₀(∅^⊙);∼^*} \& {Σ₀𝐙₀} \\
	{∅^⊙} \& {𝐙₀}
	\arrow["{≅}"', from=1-1, to=2-1]
	\arrow[from=1-1, to=1-2]
	\arrow["{≅}", from=1-2, to=2-2]
	\arrow[from=2-1, to=2-2]
      \end{tikzcd}
      \quad
\begin{tikzcd}[ampersand replacement=\&]
	{∅^⊙;∼^*} \& {𝐙₀} \\
	{∅^⊙} \& {𝐙₀}
	\arrow[from=1-1, to=1-2]
	\arrow[from=1-1, to=2-1]
	\arrow[from=2-1, to=2-2]
	\arrow[Rightarrow, no head, from=1-2, to=2-2]
      \end{tikzcd}
    \end{center}
    Because pointwise weak pullbacks are closed under (vertical)
    copairing and preserved by $Δ_{𝟚,𝐬,𝐥}$, it thus suffices to show
    that both squares are pointwise weak pullbacks.  The left square
    is one as an isomorphism in the arrow category.  The right square
    is one because it admits a cone morphism from the actual pullback,
    using reflexivity of $∼^*$ as in
    \begin{center}
\begin{tikzcd}[ampersand replacement=\&]
	{∅^⊙} \\
	\& {∅^⊙;∼^*} \& {𝐙₀} \\
	\& {∅^⊙} \& {𝐙₀\rlap{.}}
	\arrow[Rightarrow, no head, from=1-1, to=3-2]
	\arrow[from=1-1, to=2-3]
	\arrow[from=1-1, to=2-2]
	\arrow[from=2-2, to=2-3]
	\arrow[from=2-2, to=3-2]
	\arrow[from=3-2, to=3-3]
	\arrow[Rightarrow, no head, from=2-3, to=3-3]
\end{tikzcd}            
    \end{center}
\end{proof}

\section{Proof of Theorem~\ref{thm:cellular}}
\label{app:proof-cellular}

We assume some basic knowledge of familial functors.  In particular,
there is a well-known alternative characterisation in terms of
\alert{generic-free} factorisation, across which border arities are
characterised as follows.
\begin{proposition}
  In the setting of~\cref{def:border}, for any $α ∈ 𝔼𝕋$ and
  $r ∈ 𝒰'_𝟚Σ₁(1)(α)$, the border arity $𝐛ᵣ$ is isomorphic to the
  morphism $ϕ$ obtained by first factoring $r$ as
  $𝐲_α \xto{ξᵣ} 𝒰'_𝟚Σ₁(B) \xto{𝒰'_𝟚Σ₁(!)} 𝒰'_𝟚Σ₁(1)$ with $ξᵣ$
  generic, and then $ξᵣ ∘ j_{𝟚,α}$ as $F(ϕ)∘χᵣ$ with $χᵣ$ generic.
  \begin{center}
    \diag{%
      𝐲_{𝐬(α)_D} + ∑_{i ∈ n_α} 𝐲_{(𝐥^αᵢ)_V} \& 𝐲_α \\
      𝒰'_𝟚(Σ₁(A)) \& 𝒰'_𝟚(Σ₁(B)) %
    }{%
      (m-1-1) edge[labela={j_{𝟚,α}}] (m-1-2) %
      edge[labell={χᵣ}] (m-2-1) %
      (m-2-1) edge[labelb={𝒰'_𝟚(Σ₁(ϕ))}] (m-2-2) %
      (m-1-2) edge[labelr={ξᵣ}] (m-2-2) %
    }
  \end{center}
\end{proposition}

\begin{lemma}\label{lem:characfib}
  For any categories with generating cofibrations $(𝒜,𝕁)$ and $(ℬ,𝕂)$,
 a familial functor $F∶ 𝒜 → ℬ$
  preserves fibrations iff it is \alert{cellular}, in the sense that
  for all commuting squares
  \begin{equation}
    \hfil \diag{%
      C \& D \\
      F(X) \& F(Y)%
    }{%
      (m-1-1) edge[labela={k}] (m-1-2) %
      edge[labell={ξ}] (m-2-1) %
      (m-2-1) edge[labelb={F(δ)}] (m-2-2) %
      (m-1-2) edge[labelr={χ}] (m-2-2) %
    }
    \label{eq:gensquare}
  \end{equation}
  with $k ∈ 𝕂$ and $ξ$ and $χ$ generic,
  $δ$ is a cofibration (i.e., $δ ∈ \wbotrightleft{𝕁}$).
\end{lemma}
\begin{proof}
  This is a straightforward generalisation of~\citet[Lemma~7.28]{HL},
  whose proof applies \emph{mutatis mutandis}.
\end{proof}

\begin{proof}[Proof of Theorem~\ref{thm:cellular}]
  We assume given a dynamic signature $Σ₁∶ σ\Trans → σ\Trans_𝟚$ such that
  $𝒰_𝟚Σ₁$ is familial. 
  By that that 
  By Proposition~\ref{prop:ffbisim-fib},
  $Σ₁$ preserves functional flexible bisimulations if and only if $𝒰_𝟚Σ₁$ maps 
  $𝕁_σ$-fibrations to $𝕁_{𝟚,σ}$-fibrations, or equivalently, by
  Lemma~\ref{lem:characfib},
  if it is cellular.

  Clearly, $𝒰_𝟚Σ₁$ is cellular, then the border
  arities of all rules are $𝕁_σ$-cofibrations, by a straightforward
  instantiation
  of Diagram~\ref{eq:gensquare}.
  Conversely, assume that all rules are $𝕁_σ$-cofibrations and consider a
  commuting square as in Diagram~\ref{eq:gensquare}, taking $F=𝒰_𝟚Σ₁$, specialised to the
  involved sets of cofibrations:
  \begin{center}
  \diag{%
    𝐲_{𝐬(α)_D} + ∑_{i ∈ n_α} 𝐲_{(𝐥^αᵢ)_V} \& 𝐲_α \\
    F(A) \& F(B)\rlap{.} %
  }{%
    (m-1-1) edge[labela={j_{𝟚,α}}] (m-1-2) %
    edge[labell={χᵣ}] (m-2-1) %
    (m-2-1) edge[labelb={F(ϕ)}] (m-2-2) %
    (m-1-2) edge[labelr={ξᵣ}] (m-2-2) %
  }
\end{center}
The result follows by considering the rule $𝐲_α \xto{ξᵣ}F(B) \xto{F(!)} F(1)$
and exploiting uniqueness (up to isomorphism) of generic
factorisations~\cite[Remark 7.19]{HL}.
\end{proof}
\end{document}